\documentclass[a4paper,
               DIV=14, 
               titlepage=off,
               abstract=true,
               parskip=half]{scrartcl}
\usepackage{amssymb,amsmath,amsthm,amsfonts,braket}
\usepackage{mathpazo}

\newif\ifusebiblatex
\IfFileExists{biblatex.sty}{
    \usebiblatextrue
    \usepackage[backend=biber,style=alphabetic]{biblatex}
    \addbibresource{citations.bib}
}{
    \usebiblatexfalse
}

\usepackage{verbatim}
\usepackage{centernot}
\usepackage{mathtools}
\usepackage[dvipsnames]{xcolor}
\usepackage[colorlinks=true,linkcolor=blue,urlcolor=blue,citecolor=blue,anchorcolor=green]{hyperref}
\usepackage{xspace,xparse}
\usepackage{microtype}
\usepackage{tikz}
\usetikzlibrary{arrows.meta,positioning,calc}
\usepackage{authblk}

\newtheorem{thm}{Theorem}[section]
\newtheorem{prop}[thm]{Proposition}
\newtheorem{lem}[thm]{Lemma}
\newtheorem{cor}[thm]{Corollary}

\theoremstyle{definition}
\newtheorem{dfn}[thm]{Definition}
\newtheorem{eg}[thm]{Example}

\theoremstyle{remark}
\newtheorem{rmk}[thm]{Remark}

\theoremstyle{definition}

\numberwithin{equation}{section} 

\newcommand{\N}{\mathbb N}

\newcommand{\zo}{\{0,1\}}

\newcommand*{\qaczero}{\mathsf{QAC^0}}

\newcommand*{\aczero}{\mathsf{AC^0}}
\newcommand*{\toff}{\mathsf{TOFFOLI}}
\newcommand*{\fanout}{\mathsf{FANOUT}}
\newcommand*{\revor}{\mathsf{REVOR}}

\title{The power of constant-depth quantum circuits of unbounded size}
\author{Sergii Strelchuk\thanks{\texttt{Sergii.Strelchuk@cs.ox.ac.uk}}}
\author{Sathyawageeswar Subramanian\thanks{\texttt{Sathya.Subramanian@cs.ox.ac.uk}}}
\author{M\'{a}t\'{e} Weisz\thanks{\texttt{Mate.Weisz@st-annes.ox.ac.uk}}}
\affil{\small \textit{Department of Computer Science, University of Oxford, Parks Rd, Oxford OX1 3QG, United Kingdom}}

\date{}

\begin{document}

\maketitle

\begin{abstract}
Classical circuits with unbounded fan-in can compute any Boolean function in constant depth when their size is unrestricted. We ask whether removing the restrictions on circuit size and ancillary qubits also allows quantum circuits built from arbitrary single-qubit gates and generalised Toffoli gates to implement every unitary in constant depth.

We give exact constant-depth constructions for arbitrary permutations of computational basis states, diagonal unitaries and the preparation of arbitrary pure states. These connect quantum state preparation to reversible classical computation and the preparation of probability distributions. With fanout in the gate set, they use exponentially many gates and ancillary qubits and return all ancillary qubits to zero. Replacing fanout by an exact circuit over the original gate set preserves constant depth, although the size bounds can become doubly exponential.

The implementation of arbitrary unitaries in constant depth remains open. We give equivalent formulations in terms of copying the vectors of a specified orthonormal basis, extracting their labels and implementing restricted families of unitaries. We also reduce arbitrary unitary implementation to that of traceless unitary involutions using one additional clean qubit. With adaptive measurements, gate teleportation gives depth proportional to the level of a gate in the Clifford hierarchy.

Towards arbitrary unitary implementation in constant depth, we study port-based teleportation. For input dimension $d$ and $M\geq d^2-1$ ports, we construct a unitary circuit of depth $O(\sqrt d)$, independent of $M$ and including resource preparation and port selection, with entanglement fidelity at least $(1-(d^2-1)/(2M))^2$. Thus, at fixed $d$, the approximation can be made arbitrarily accurate without increasing depth. Whether the dependence on $d$ can also be removed remains open.
\end{abstract}

\clearpage
\tableofcontents
\clearpage

\section{Introduction}

The classical circuit class $\aczero$ is defined by polynomial-size circuit families of constant depth over unbounded fan-in AND and OR gates and NOT
gates~\cite[Section~4.5]{odonnell2014analysis}. Notably, parity lies outside
this class~\cite{furst1984parity}. The restriction to polynomial size is essential to the limitations of classical circuits of constant depth with unbounded fan-in, since with unrestricted size, circuits over the same gate set can compute every Boolean function in constant depth. In particular, every function on $n$ bits has a disjunctive normal form (DNF) with at most $2^n$ conjunctions. These conjunctions can be evaluated
in parallel and combined by a single OR gate, giving constant depth and $O(2^n)$ gates~\cite[Section~4.1]{odonnell2014analysis}. 

While this does not place every Boolean function in $\aczero$, it demonstrates that if efficiency is set aside, constant depth alone imposes no restriction on which Boolean functions can be computed.

We investigate a quantum counterpart of this phenomenon, focusing on circuits built from arbitrary single-qubit gates and generalised Toffoli gates with an unrestricted number of controls. When the size is restricted to be polynomial, circuits over this gate set form the class $\qaczero$. 

Suppose instead that we relax the restrictions on both circuit size and ancillary space. Can every unitary then be implemented exactly in constant depth? We approach this question through four increasingly general tasks:
\begin{enumerate}
    \item computing membership in any set $L\subseteq\zo^n$,
    \item implementing any permutation of computational basis states,
    \item preparing any pure quantum state,
    \item implementing any unitary on every input state.
\end{enumerate}

We give general constructions showing that unrestricted size and ancillary space suffice for the first three tasks, but our constructions do not extend directly to implementing arbitrary unitaries. The difficulty is no longer to produce an arbitrary output state from a fixed input, but to prescribe the action on every input state simultaneously while preserving unitarity, making the fourth task of implementing unitaries qualitatively different from the first three. Nevertheless, we provide explicit constructions for common families of unitaries, and provide several equivalent characterisations of necessary and sufficient conditions for arbitrary unitaries.

Classically, removing the size restriction reduces constant-depth universality to the elementary DNF construction. The quantum setting points to a broader question: once efficiency is set aside, what ultimately limits constant-depth computation? 

\subsection{Main results}

The first two tasks admit a common solution. We construct a reversible
encoding that maps each $n$-bit string to its $2^n$-dimensional indicator
vector. The construction tests the input against every possible string in
parallel, while supplying the input to the separate tests and subsequently
clearing all additional registers. Composing two such encodings implements
an arbitrary permutation of bitstrings, and hence an arbitrary permutation
of computational basis states. The same construction also computes any
Boolean function while retaining its input, through the reversible map
\begin{equation}
    (x,y)\longmapsto(x,y\oplus f(x)).
\end{equation}
Retaining the input is necessary in general, since $x\mapsto f(x)$ is not
reversible when $f$ is noninjective. In the quantum setting, the indicator
encoding also lets us implement an arbitrary diagonal unitary: we compute
the indicator of the input, apply all possible phases to the corresponding
indicator qubits in parallel, and uncompute the indicators.

For the third task, we first give a probabilistic classical construction
that prepares an arbitrary probability distribution. We sample bits
independently and retain only the first $1$, using its position, together
with the all-zero outcome, to encode the desired distribution. We then make
this construction coherent to prepare arbitrary pure quantum states.
Classically, the bits following the first $1$ can simply be erased;
quantumly, the corresponding qubits must instead be returned to $\ket 0$
without destroying the superposition over outputs. We achieve this by
applying inverse rotations to their known states. The resulting
construction preserves the amplitudes associated with the possible
positions of the first $1$, after which the inverse indicator encoding
converts these positions into computational basis states and a diagonal
unitary restores the required phases.

With fanout included in the gate set, all of these constructions have
constant depth. Arbitrary permutations and diagonal unitaries use
$O(n2^n)$ gates and $O(n2^n)$ bits or qubits. Arbitrary probability
distributions and pure quantum states use $O(n2^n)$ elementary gates and
$O(4^n)$ bits or qubits. In every case, all ancillary registers are returned
to zero.

Fanout is not part of our original quantum gate set. In Section~4, we
combine a result of Grier, Morris and Wu~\cite{grier2026mathsfqac0containsmathsftc0with}
with Rosenthal's parity construction and reduction from parity to
fanout~\cite{rosenthal2020boundsqac0complexityapproximating} to implement
fanout exactly in constant depth using only single-qubit and generalised
Toffoli gates. Consequently, each construction above has an exact
constant-depth implementation over our original gate set. Replacing
fanout in this way can, however, increase the size and ancillary-space
bounds to doubly exponential in $n$. Thus unrestricted circuit size and
ancillary space suffice for the first three tasks.

The fourth task remains open. The distinction from state preparation is
fundamental: preparing a state specifies the action of a circuit on one
fixed input, whereas implementing a unitary specifies its action on every
input. Section~5 studies this remaining problem through several reductions.
For a unitary diagonal in an orthonormal basis
$\{\ket{\psi_1},\ldots,\ket{\psi_N}\}$, we show that it suffices to
implement any of several operations that expose enough information about
the unknown basis vector $\ket{\psi_i}$. These include producing sufficiently
many copies of $\ket{\psi_i}$, producing a suitable permutation of all
basis vectors, and coherently extracting the computational basis label
$i$. Conversely, if arbitrary unitaries admit constant-depth
implementations, each of these operations does as well. They therefore
give equivalent formulations of arbitrary unitary implementation in our
unrestricted-resource model.

We obtain two further structural reductions. First, the normal form of
Idel and Wolf~\cite{idel2015sinkhorn} reduces arbitrary unitary
implementation to unitaries whose rows and columns all sum to $1$, since
the required diagonal factors are already available in constant depth.
Second, arbitrary unitary implementation reduces, using one additional
clean qubit, to implementing traceless unitary involutions: for every
unitary $U$, the block unitary
\begin{equation}
    \begin{pmatrix}
        0 & U^\dagger \\
        U & 0
    \end{pmatrix}
\end{equation}
is traceless, squares to the identity, and can be used to apply $U$ while
returning the additional qubit to zero. We also give an alternative
description of the circuit model: the same families of unitaries are
implementable using a constant number of layers, each consisting of
Hadamard gates, an arbitrary permutation of computational basis states,
or an arbitrary diagonal unitary.

Allowing adaptive intermediate measurements gives another route to unitary
implementation. Using gate teleportation~\cite{gottesman1999teleportation},
we show that every gate at level $\ell$ of the Clifford hierarchy has an
adaptive implementation of depth $O(\ell)$. The teleportation correction
drops by one level of the hierarchy at each round. More generally, gate
teleportation reduces arbitrary unitary implementation to the same family
of traceless unitary involutions: apart from the identity branch, the
required correction is a unitary conjugate of a nonidentity Pauli and is
therefore a traceless involution.

Finally, in Section~\ref{sec-pbt-fixed}, we use port-based teleportation
(PBT) to approach arbitrary unitary implementation without corrections
that depend on the unitary~\cite{pbt-ishizaka2008,pbt-wills2024}. For input
dimension $d$ and $M\geq d^2-1$ ports, we construct a unitary circuit using
only single-qubit and generalised Toffoli gates, with depth $O(\sqrt d)$
and entanglement fidelity
\begin{equation}
    F_{\mathrm e}
    \geq
    \left(1-\frac{d^2-1}{2M}\right)^2.
\end{equation}
The depth includes resource preparation and output selection, requires no
intermediate measurements, and is independent of $M$. Replacing the
maximally entangled resource states by copies of the Choi state of a
unitary gives an approximate implementation of that unitary. Hence, for
every fixed input dimension, the approximation can be made arbitrarily
accurate by increasing the number of ports without increasing the depth.
Whether the dependence on $d$ can also be removed remains open.

Figures~\ref{figure-preparation} and~\ref{figure-unitaries} summarise the
constructions and reductions.
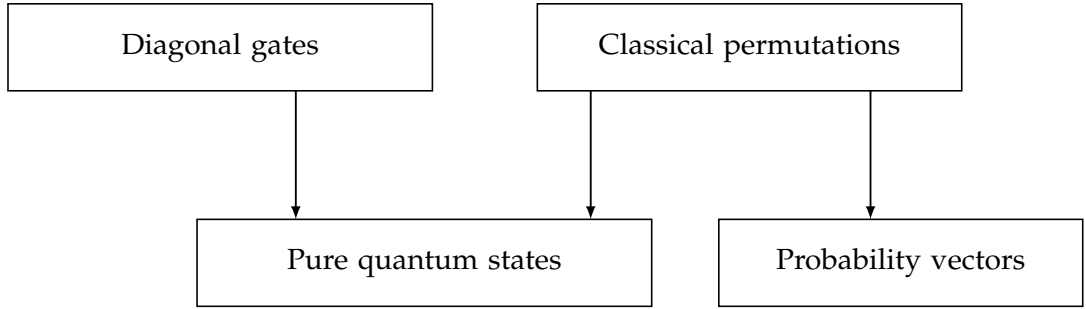
\begin{figure}[htbp]
    \centering
\begin{tikzpicture}[
    font=\normalfont\normalsize,
    every node/.style={align=center},
    box/.style={draw=black,line width=0.6pt,minimum height=1.15cm,
                inner xsep=8pt,inner ysep=8pt},
    arrow/.style={->,>=latex,line width=0.7pt}
]
\node[box,text width=5.05cm] (diagonal) at (-4.0,2.85)
    {Diagonal gates};
\node[box,text width=5.05cm] (permutations) at (3.0,2.85)
    {Classical permutations};
\node[box,text width=5.45cm] (states) at (-1.3,0)
    {Pure quantum states};
\node[box,text width=4.25cm] (probabilities) at (5.0,0)
    {Probability vectors};
\draw[arrow] (-3.0,2.275) -- (-3.0,0.575);
\draw[arrow] (0.9,2.275) -- (0.9,0.575);
\draw[arrow] (4.6,2.275) -- (4.6,0.575);
\end{tikzpicture}
    \caption{Ingredients used in the preparation constructions of Propositions~\ref{Probability distribution vector in constant depth} and~\ref{Any state in constant depth}. The arrows indicate parts of these constructions. Probability preparation also uses single-bit stochastic gates, and quantum preparation also uses single-qubit rotations. The final conversion from indicators to binary labels uses the reversible circuit of Lemma~\ref{permutation to indicator in constant depth}.}
    \label{figure-preparation}
\end{figure}

\begin{figure}[htbp]
    \centering
    \begin{tikzpicture}[
    font=\normalfont\normalsize,
    every node/.style={align=center},
    box/.style={draw=black,line width=0.6pt,minimum height=1.9cm,
                inner xsep=7pt,inner ysep=8pt,text width=4.05cm},
    arrow/.style={->,>=latex,line width=0.7pt}
]
\node[anchor=west,font=\normalfont\bfseries,align=left]
    at (-7.28,6.0) {(a) Exact implementation in constant depth};
\node[box] (labels) at (-5.0,3.0)
    {Basis labels\\[3pt]
     {\small Proposition~\ref{basis labels} and\\Corollary~\ref{labels imply unitaries}}};
\node[box] (involutions) at (0,4.5)
    {Traceless involutions\\[3pt]
     {\small Direct unitary reduction\\
      after Proposition~\ref{traceless involutions}}};
\node[box] (copying) at (5.0,3.0)
    {Copying basis vectors\\[3pt]
     {\small Proposition~\ref{copying basis vectors}}};
\node[box,text width=4.5cm,minimum height=1.35cm] (unitary) at (0,0.65)
    {Arbitrary unitary\\implementation\\[3pt]
     {\small Open problem}};
\node[box,text width=4.55cm] (sums) at (-4.4,-2.0)
    {Unitaries with row\\and column sums $1$\\[3pt]
     {\small Proposition~\ref{row and column sums}}};
\node[box,text width=4.55cm] (lists) at (4.4,-2.0)
    {Permuted lists\\of basis vectors\\[3pt]
     {\small Proposition~\ref{permuted basis vectors}}};
\draw[arrow] (labels.south) -- (unitary.north west);
\draw[arrow] (involutions.south) -- (unitary.north);
\draw[arrow] (copying.south) -- (unitary.north east);
\draw[arrow] (sums.north) -- (unitary.south west);
\draw[arrow] (lists.north) -- (unitary.south east);

\draw[black!35,line width=0.45pt] (-7.28,-3.55) -- (7.28,-3.55);
\node[anchor=west,font=\normalfont\bfseries,align=left]
    at (-7.28,-4.05) {(b) Approximation using port-based teleportation with $M$ ports};

\node[box,minimum height=2.05cm] (choi) at (-5.0,-5.85)
    {Choi states of $U$\\[3pt]
     {\small Constant-depth preparation\\
      Theorem~\ref{Any state in constant depth}}};
\node[box,minimum height=2.05cm] (pbt) at (0,-5.85)
    {Port-based teleportation\\[3pt]
     {\small Depth $O(\sqrt d)$\\
      Theorem~\ref{pbt-main}}};
\node[box,minimum height=2.05cm] (approximate) at (5.0,-5.85)
    {Approximate unitary\\implementation\\[3pt]
     {\small Depth $O(\sqrt d)$\\Independent of $M$}};
\draw[arrow] (choi.east) -- (pbt.west);
\draw[arrow] (pbt.east) -- (approximate.west);

\end{tikzpicture}
\caption{Exact reductions and approximation via port-based
teleportation. (a) Each arrow gives a sufficient condition
for exact constant-depth implementation of arbitrary
unitaries. The traceless-involution arrow uses the direct
unitary reduction following
Proposition~\ref{traceless involutions}, with one additional
clean qubit. (b) Preparing copies of the Choi state of $U$
and applying port-based teleportation gives an approximate
implementation of $U$. The quality of approximation is measured by the entanglement fidelity relative
to $U$, and the depth includes resource preparation and
output selection.
The number of ancillary qubits are
unrestricted in both panels.}
    \label{figure-unitaries}
\end{figure}

\subsection{Related work}
As mentioned briefly in the introduction, the broader context for our question may be understood by considering the case of quantum parity and fanout. On computational basis states the $n+1$-qubit quantum parity gate replaces its target bit by its XOR with all control bits and leaves the controls unchanged. This action extends linearly to arbitrary quantum states. Moore showed that a quantum version of fanout and parity can be reduced to one another in constant depth using single-qubit gates~\cite{moore1999fanout}.
Further expanding on the power of such circuits,  H{\o}yer and {\v S}palek subsequently showed that adding fanout to the gate set permits polynomial-size, constant-depth quantum circuits to approximate a range of counting and arithmetic operations~\cite{hoyer2005fanout}.

Whether fanout can itself be implemented in constant depth from single-qubit and generalised Toffoli gates depends on the available ancillary space. Indeed, Fang et al.\ proved that parity and fanout require at least logarithmic depth over this gate set when only a constant number of ancillary qubits is available~\cite{fang2006fanout}. This restriction on ancillary
qubits leaves open the possibility that additional space, together with sufficiently large circuits, permits constant depth implementations.

The focus of our investigation is therefore whether relaxing restrictions on circuit size and ancillary space can compensate for operations that are otherwise unavailable at constant depth.

Rosenthal~\cite{rosenthal2026querydepth} proved that arbitrary pure states can be prepared exactly in constant depth using exponentially many ancillary qubits when fanout is included in the gate set. Gretta, Gupta and Joshi~\cite[Theorem~1.1]{gretta2026logarithmic} give exact constant-depth preparation using $2^{O(n)}$ gates and ancillary qubits with only single-qubit and generalised Toffoli gates. Both constructions return the ancillary qubits to zero. Consequently, the potentially doubly exponential size obtained by replacing fanout in our construction is not the best known bound for arbitrary state preparation over this gate set. Our construction instead gives a direct preparation procedure that makes explicit the connection between classical probability preparation and quantum state preparation.

Much smaller circuits suffice for restricted families of states. A Dicke state is the normalised equal superposition of computational-basis states of a fixed Hamming weight. Gretta, Gupta and Joshi~\cite{gretta2026dicke} give constant-depth constructions using only single-qubit and generalised Toffoli gates for Dicke states whose weight is bounded by a polynomial in $\log n$. They also give constant-depth constructions with fanout for arbitrary symmetric states, i.e., states invariant under every permutation of their qubits. These works focus on identifying families of states that admit polynomial-resource preparation, whereas our results concern arbitrary states and permit the circuit size to grow without restriction.

\subsection{Circuit conventions}
\label{circuit-conventions}

We work with reversible classical, probabilistic classical and quantum circuits. In Section~2 we use generalised Toffoli and fanout gates as elementary gates. In Section~3 we also allow arbitrary single-bit stochastic gates. In Sections~4 and~5 we use arbitrary single-qubit unitaries and generalised Toffoli gates. We denote this quantum model by unbounded size $\qaczero$ and write $\mathsf{QAC^0_f}$ when fanout is also included in the gate set. In both quantum models we remove the polynomial size restriction.

Each layer consists of gates acting on disjoint sets of bits or qubits, including controls and targets. Depth is the number of layers, size is the number of elementary gates, and width is the total number of bits or qubits, including ancillary registers. We count each elementary gate once regardless of its number of inputs. A fanout with any number of targets counts as one gate and occupies one layer when it is included in the gate set. Otherwise we implement fanout by a circuit over the chosen gate set and count the gates, depth and ancillary qubits of that implementation.

For each constant-depth construction, we require a single depth bound that holds for every input length and every target covered by the result. The circuit, its width and its gate parameters may depend on the target. We treat each chosen single-bit stochastic map or single-qubit unitary as an exact elementary gate. We make no assumption that the circuit or its parameters can be found efficiently.

We call an ancillary qubit clean if it is initialised in $\ket{0}$. To implement an $n$-qubit unitary $U$ we allow a circuit $C$ on $n+c$ qubits for some $c\geq0$ and require
\[
    C\bigl(\ket{\psi}\ket{0}^{\otimes c}\bigr)
    =U\ket{\psi}\ket{0}^{\otimes c}
    \qquad\text{for every input state }\ket{\psi}.
\]
The first register contains the data and the second contains the ancillary qubits. Returning the ancillary qubits to zero allows us to compose these implementations. State preparation specifies the output only on the all-zero input and also requires all ancillary qubits to return to zero. In reversible and probabilistic classical circuits the input register contains the specified bitstring or probability vector. All other bits start in zero and every bit outside the output register ends in zero.

Input and output registers may occupy different locations. The indicator encoding in Section~2 uses this convention. 
Some intermediate maps in Section~5 produce additional states in ancillary registers. We return these registers to zero by applying the inverse of the map that produced them. Each construction specifies how to do this while preserving the required output. We introduce adaptive measurements only in the final part of Section~5, where we require the corrected quantum output to agree with the target unitary and allow measurement records to be discarded.

\section{Reversible classical circuits}
\label{section-reversible}

We begin by constructing constant-depth reversible circuits for arbitrary permutations of bitstrings. These constructions will also apply to quantum circuits, where the same gates permute computational basis states. We first specify the gates and their action on selected bits, using the circuit conventions from Section~\ref{circuit-conventions}. We write $[m]=\{1,\ldots,m\}$.

\begin{dfn}
    A reversible classical gate on $n$ bits is a permutation $\pi:\zo^n\rightarrow\zo^n$. On input $x=x_1\ldots x_n$, it outputs $\pi(x)$.
\end{dfn}

\begin{rmk}
    We write $P_n$ for the group of permutations of $n$-bit strings, so $P_n\cong S_{2^n}$.
\end{rmk}

\begin{eg}
    The generalised Toffoli gate $\toff_n$ has $n-1$ control bits and one target bit. It replaces the target by its XOR with
the AND of the controls:
    \begin{equation}
        \toff_n(x_1\ldots x_n)
        :=x_1\ldots x_{n-1}\bigl(x_n\oplus(x_1\wedge\cdots\wedge x_{n-1})\bigr).
    \end{equation}
    It exchanges $1^{n-1}0$ and $1^n$ and leaves every other bitstring unchanged. The first $n-1$ bits are its controls and the last bit is its target. The controls may be listed in any order.

    By convention, $\toff_1=\mathsf{NOT}$.
\end{eg}

\begin{eg}
    The $n$-bit $\fanout$ gate replaces each of the last $n-1$ bits by its XOR with the first bit:
    \begin{equation}
        \fanout_n(x_1\ldots x_n)
        :=x_1(x_2\oplus x_1)\ldots(x_n\oplus x_1).
    \end{equation}
    The first bit is its control and the remaining bits are its targets. The targets may be listed in any order. Applying the gate twice gives the identity.

    By convention, $\fanout_1$ is the identity.
\end{eg}

\begin{eg}
\label{example-revor}
    The $n$-bit $\revor$ gate replaces the last bit by its XOR with the OR of the first $n-1$ bits:
    \begin{equation}
        \revor_n(x_1\ldots x_n)
        :=x_1\ldots x_{n-1}\bigl(x_n\oplus(x_1\lor\cdots\lor x_{n-1})\bigr).
    \end{equation}
    The first $n-1$ bits are its controls and the last bit is its target. Applying this gate twice also gives the identity. We take an empty OR to be zero, so $\revor_1$ is the identity.

    For $n\geq2$, we implement this gate in three layers using Toffoli and NOT gates. We negate all controls, apply $\toff_n$, and then negate all controls and the target. These operations restore the controls to their original values and replace the target by
    \begin{equation}
    \begin{split}
        x_n\oplus(\neg x_1\wedge\cdots\wedge\neg x_{n-1})\oplus1
        &=x_n\oplus(x_1\lor\cdots\lor x_{n-1}).
    \end{split}
    \end{equation}
    The implementation uses $2n$ elementary gates.
\end{eg}

\begin{dfn}
    An embedded reversible gate on $n$ bits consists of a gate $\pi\in P_m$ and an injective location function $f:[m]\rightarrow[n]$. The location $f(j)$ holds the $j$th input and output of $\pi$, while the remaining bits are unchanged. Writing $\pi_j$ for the $j$th output bit of $\pi$, its action sends $x_1\ldots x_n$ to $y_1\ldots y_n$, where
    \begin{equation}
        y_i=
        \begin{cases}
            \pi_j(x_{f(1)}\ldots x_{f(m)})&\text{if }i=f(j)\text{ for some }j\in[m],\\
            x_i&\text{otherwise}.
        \end{cases}
    \end{equation}
\end{dfn}

We denote this embedded gate by $(\pi,f)$.

\begin{rmk}
    Injectivity of $f$ ensures that each output is assigned to a distinct bit. For a Toffoli gate, we need only specify the set of control locations and the target location, since permuting the controls does not change its action. Similarly, a fanout gate is specified by its control location and the set of target locations.
\end{rmk}

\begin{dfn}
    Let $\mathcal G\subseteq\bigcup_{i\in\N}P_i$ be a set of elementary gates. A layer $L$ of embedded $\mathcal G$-gates on $n$ bits is a set
    \begin{equation}
        L=\{(\pi_1,f_1),\ldots,(\pi_k,f_k)\},
    \end{equation}
    where every $\pi_i$ belongs to $\mathcal G$ and the ranges of the $f_i$ are pairwise disjoint. The disjointness condition allows all gates in the layer to act in parallel.
\end{dfn}

\begin{rmk}
    Gates in a layer act on disjoint sets of bits and hence commute. Their composition gives the action of the layer in any order.
\end{rmk}

\begin{dfn}
    A reversible classical circuit $C$ on $n$ bits over a gate set $\mathcal G$ is a tuple $(L_1,\ldots,L_d)$ of layers of embedded $\mathcal G$-gates. We apply the layers in the order $L_1,\ldots,L_d$. The circuit has depth $d$ and width $n$.
\end{dfn}

\begin{rmk}
    Identifying circuits and layers with their actions, we have $C=L_d\circ\cdots\circ L_1$.
\end{rmk}

For the rest of this section, we use Toffoli and fanout gates of arbitrary arity as elementary gates.

\begin{dfn}
    A family of reversible classical circuits $(C_m)_{m\in\N}$ over the gate set
    \begin{equation}
        \mathcal G=\{\toff_n,\fanout_n:n\geq1\}
    \end{equation}
    is a family of constant-depth circuits over Toffoli and fanout gates, abbreviated CDTF, if a single constant bounds the depth of every $C_m$.
\end{dfn}

\begin{rmk}
    We may add identity layers so that all circuits in the family have the same depth.
\end{rmk}

\begin{dfn}
    A circuit $C$ on $n$ bits realises a permutation $\pi\in P_m$ at distinct locations $(b_1,\ldots,b_m)\in[n]^m$ if it agrees with the embedded gate $(\pi,f)$, where $f(i)=b_i$, on every input that is zero outside these locations. Thus the circuit applies $\pi$ to the selected bits and returns every additional bit to zero.
\end{dfn}

\begin{dfn}
\label{definition-classical-implementation}
    Let $S\subseteq\zo^k$ and $k,m\leq n$. A circuit $C$ on $n$ bits implements an injective function $f:S\rightarrow\zo^m$ if, for chosen input and output locations,
    \begin{equation}
        C(x0^{n-k})=f(x)0^{n-m}\qquad\forall x\in S.
    \end{equation}
    We write each side with its corresponding data register first. The input and output locations may differ, and every bit outside the output register returns to zero. Since a reversible circuit is a permutation, $f$ must be injective on $S$. Conversely, any such injective map extends to a permutation of all $n$-bit strings: we match the remaining inputs bijectively with the remaining outputs.
\end{dfn}

To realise any permutation, we will rely on Lemma \ref{permutation to indicator in constant depth}, which shows how to map any ordering of $\zo^n$ into a fixed list of strings. More concretely, we will use a list of the strings of length $2^n$ with Hamming weight 1. Using the lemma twice will allow us to go from any ordering to any other.

Our construction encodes each bitstring by an indicator vector, with one entry for each possible input. Exactly one entry is 1, identifying the input. We first compute these entries into an additional register, then use them to clear the original input register. Reversing the circuit recovers the bitstring from its indicator vector. We will implement an arbitrary permutation by composing a circuit from Lemma \ref{permutation to indicator in constant depth} with the inverse of another.

\begin{lem}
\label{permutation to indicator in constant depth}
    Let $n\geq1$, let $m=2^n$, and let $x^1,\ldots,x^m$ be any ordering of $\zo^n$. The map
    \begin{equation}
        f:\zo^n\rightarrow\zo^m,\qquad
        x\longmapsto\mathbb{1}_{x^1}(x)\ldots\mathbb{1}_{x^m}(x)
    \end{equation}
    can be implemented by a reversible classical circuit over Toffoli and fanout gates with depth at most $10$, width $(n+1)m$ and size $O(nm)$. Here $\mathbb{1}_{x^i}(x)$ is one if $x=x^i$ and zero otherwise. The total number of gate inputs is also $O(nm)$.
\end{lem}

\begin{proof}

We construct the circuit in two parts: first we compute the indicator vector, then we use it to clear the input. The arrays below illustrate the construction for $n=3$, with the ordering
\begin{equation}
    x^1=110, x^2=010, x^3=001, x^4=000, x^5=101, x^6=100, x^7=111, x^8=011.
\end{equation}

For general $n$, we arrange $mn$ bits in an $m\times n$ array $A$. Initially, its first row contains the input $x$, with $x_j$ at position $A_{1j}$, and all remaining bits are zero. We use a column $B$ of $m$ further bits, also initially zero, to store the indicators. The superscripts on the arrays label the displayed stages, each of which may involve several circuit layers.

\begin{equation}
   A^{(0)}= \begin{array}{|c|c|c|}
        \hline
         x_1 & x_2 & x_3 \\
         \hline
         0 & 0 & 0 \\
         \hline
         0 & 0 & 0 \\
         \hline
         \vdots & \vdots & \vdots \\
         \hline
         0 & 0 & 0  \\
         \hline
    \end{array}
    \hspace{1cm}
    B^{(0)}=\begin{array}{|c|}
        \hline
          0 \\
          \hline
          0 \\
          \hline
          0 \\
          \hline
          \vdots \\
          \hline
          0 \\
          \hline
    \end{array}
\end{equation}

We copy the input into every row of $A$ by applying a $\fanout_m$ gate to each column, using the first bit of the column as its control. The columns are disjoint, so these gates form one layer. The register $B$ remains zero.

\begin{equation}
   A^{(1)}= \begin{array}{|c|c|c|}
        \hline
         x_1 & x_2 & x_3 \\
         \hline
         x_1 & x_2 & x_3 \\
         \hline
         x_1 & x_2 & x_3\\
         \hline
         \vdots & \vdots & \vdots\\
         \hline
         x_1 & x_2 & x_3\\
         \hline
    \end{array}
\end{equation}

To test whether $x=x^i$ in each row $i$, we apply a $\mathsf{NOT}$ gate at every position $A_{ij}$ for which $(x^i)_j=0$. After this layer, all entries in row $i$ are 1 exactly when $x=x^i$. Since the list contains every bitstring once, there is precisely one such row, which we denote by $s$. This layer leaves $B$ unchanged.

    \begin{equation}
   A^{(2)}= \begin{array}{|c|c|c|}
        \hline
         x_1 & x_2 & \neg x_3 \\
         \hline
         \neg x_1 & x_2 & \neg x_3 \\
         \hline
         \neg x_1 & \neg x_2 & x_3\\
         \hline
         \neg x_1 & \neg x_2 & \neg x_3\\
         \hline
         x_1 & \neg x_2 & x_3\\
         \hline
         x_1 & \neg x_2 & \neg x_3\\
         \hline
         x_1 & x_2 & x_3\\
         \hline
         \neg x_1 & x_2 & x_3\\
         \hline
    \end{array}
\end{equation}

We now write the result of each test into $B_i$ by applying a $\toff_{n+1}$ gate with row $i$ as its controls and $B_i$ as its target. These gates act on disjoint sets of bits, so they form one layer. The register $B$ then contains the indicator vector: $B_s=1$ and all its other entries are zero. The register $A$ is unchanged.

    \begin{equation}
    B^{(3)}=\begin{array}{|c|}
        \hline
          \mathbb{1}_{x^1}(x) \\
          \hline
          \mathbb{1}_{x^2}(x) \\
          \hline
          \mathbb{1}_{x^3}(x) \\
          \hline
          \mathbb{1}_{x^4}(x) \\
         \hline
          \mathbb{1}_{x^5}(x) \\
         \hline
          \mathbb{1}_{x^6}(x) \\
         \hline
          \mathbb{1}_{x^7}(x) \\
          \hline
          \mathbb{1}_{x^8}(x) \\
          \hline
    \end{array}
\end{equation}

We reverse the NOT layer and then the column fanout layer to restore $A$ to its initial state. Its first row again contains $x=x^s$, with zero in every other entry. Neither layer acts on $B$, so the indicator vector is preserved.

\begin{equation}
   A^{(4)}= \begin{array}{|c|c|c|}
        \hline
         x_1 & x_2 & x_3 \\
         \hline
         0 & 0 & 0 \\
         \hline
         0 & 0 & 0 \\
         \hline
         \vdots & \vdots & \vdots\\
         \hline
         0 & 0 & 0 \\
         \hline
    \end{array}
\end{equation}

For the second part of the construction, we use the indicators to clear the input. We denote the first row of $A$ by $C$ and arrange the register $B$, together with $m(n-1)$ of the zero bits in $A$, as an $m\times n$ array $D$. The first column of $D$ is $B$ and all its other entries are zero. This is a relabelling of existing bit locations and requires no gates. We have enough zero bits because
\begin{equation}
 n(m-1)-m(n-1)=m-n\geq0,\qquad m=2^n,\quad n\geq1.
\end{equation}
The remaining $m-n$ bits stay zero during the rest of the construction. The register $C$ still contains $x$.

\begin{equation}
    D^{(0)}=\begin{array}{|c|c|c|}
        \hline
           \mathbb{1}_{x^1}(x) & 0 & 0\\
          \hline
          \mathbb{1}_{x^2}(x)  & 0 & 0\\
          \hline
          \mathbb{1}_{x^3}(x)  & 0 & 0\\
          \hline
          \vdots &\vdots & \vdots \\
          \hline
           \mathbb{1}_{x^8}(x)  & 0 & 0\\
          \hline
    \end{array}
\end{equation}

To clear the $n$ bits of $C$ in parallel, we first make $n$ copies of each indicator. We apply $\fanout_n$ to every row of $D$, using its first bit as the control. These gates form one layer, since the rows are disjoint. The register $C$ is unchanged.

\begin{samepage}
\begin{equation}
    D^{(1)}=\begin{array}{|c|c|c|}
        \hline
           \mathbb{1}_{x^1}(x) & \mathbb{1}_{x^1}(x) & \mathbb{1}_{x^1}(x)\\
          \hline
          \mathbb{1}_{x^2}(x)  & \mathbb{1}_{x^2}(x) & \mathbb{1}_{x^2}(x)\\
          \hline
          \mathbb{1}_{x^3}(x)  & \mathbb{1}_{x^3}(x) & \mathbb{1}_{x^3}(x)\\
          \hline
          \vdots &\vdots & \vdots \\
          \hline
           \mathbb{1}_{x^8}(x)  & \mathbb{1}_{x^8}(x) & \mathbb{1}_{x^8}(x)\\
          \hline
    \end{array}
\end{equation}
\end{samepage}

For each $j\in[n]$, let $S_j:=\{i\in[m]:(x^i)_j=1\}$. We apply a $\revor$ gate with the bits $D_{ij}$ for $i\in S_j$ as its controls and $C_j$ as its target. Since exactly one indicator is 1, the OR of these controls recovers the input bit $x_j$:
\begin{equation}
 \bigvee_{i\in S_j}\mathbb{1}_{x^i}(x)=x_j.
\end{equation}
Each gate thus replaces $C_j=x_j$ by $x_j\oplus x_j=0$. The gates for different $j$ use different columns of $D$ and distinct targets in $C$, so they have disjoint supports. Using Example~\ref{example-revor}, we implement all of them in three layers. The register $D$ is unchanged.

\begin{equation*}
   C^{(2)}= \begin{array}{|c|c|c|}
        \hline
         0 & 0 & 0 \\
         \hline
         \end{array}
\end{equation*}

Finally, we reverse the fanout layer on the rows of $D$ to clear the additional indicator copies. The first column, which is the original register $B$, retains the indicators, and every other bit returns to zero. The register $C$ remains zero.

\begin{equation}
    D^{(3)}=\begin{array}{|c|c|c|}
        \hline
           \mathbb{1}_{x^1}(x) & 0 & 0\\
          \hline
          \mathbb{1}_{x^2}(x)  & 0 & 0\\
          \hline
          \mathbb{1}_{x^3}(x) & 0 & 0\\
          \hline
          \vdots &\vdots & \vdots \\
          \hline
           \mathbb{1}_{x^8}(x)  & 0 & 0\\
          \hline
    \end{array}
\end{equation}
We have used five layers to compute the indicators and restore $A$, followed by five layers to clear the input and the additional indicator copies. The circuit has depth at most $10$ and uses exactly $(n+1)m$ bits. Counting the elementary gates in these layers gives size $O(nm)$, and the total number of gate inputs is also $O(nm)$.
\end{proof}

We now implement an arbitrary permutation by composing an indicator circuit with the inverse of another. We choose the two orderings of the bitstrings so that the composition applies the required permutation and returns all additional bits to zero.

\begin{prop}
\label{Any permutation in constant depth}
    CDTF circuits are universal for reversible classical computation. More precisely, every permutation $\pi\in P_n$ can be realised by a circuit over Toffoli and fanout gates of depth at most $20$, width $(n+1)2^n$ and size $O(n2^n)$.
\end{prop}

\begin{proof}
    Let $m=2^n$ and let $x^1,\ldots,x^m$ be the lexicographical ordering of $\zo^n$. By Lemma~\ref{permutation to indicator in constant depth}, we obtain a circuit $C_1$ of depth at most $10$ that implements
    \begin{equation}
        x\longmapsto\mathbb{1}_{x^1}(x)\ldots\mathbb{1}_{x^m}(x).
    \end{equation}
    We apply the same lemma to the ordering $\pi(x^1),\ldots,\pi(x^m)$ to obtain a circuit $C_2$ with the same bounds that implements
    \begin{equation}
        x\longmapsto\mathbb{1}_{\pi(x^1)}(x)\ldots\mathbb{1}_{\pi(x^m)}(x).
    \end{equation}
    We choose the same input and indicator locations for both circuits. On input $x^j$, $C_1$ produces the indicator vector with its 1 in position $j$ and clears every other bit. Since $C_2$ maps $\pi(x^j)$ to this same indicator vector, $C_2^{-1}\circ C_1$ sends $x^j$ to $\pi(x^j)$ and returns all additional bits to zero.

    Toffoli and fanout gates are their own inverses, so we reverse $C_2$ without changing its depth or gate set. The composition has depth at most $20$, uses the same $(n+1)2^n$ bits, and has $O(n2^n)$ elementary gates. Its total number of gate inputs is also $O(n2^n)$.
\end{proof}

We will also require the following construction of a reversible classical circuit that we will use as a subroutine in later constructions in Sections \ref{sec:prob-classical-ckts} and \ref{quantum circuits}.

\begin{lem}
\label{prefix-values-linear-gates}
Let $N\geq1$ and $T=N(N-1)/2$. There is a reversible classical circuit
$C$ over generalised Toffoli and fanout gates such that
\begin{equation}
    C(x,z,0^T)=(x,z\oplus s(x),0^T)
    \qquad\text{for every }x,z\in\zo^N.
\end{equation}
The circuit has depth at most five, width $2N+T$ and at most
$6(N-1)$ elementary gates. The same circuit implements this map
exactly on quantum states and returns all $T$ ancillary qubits to zero.
\end{lem}

\begin{proof}
For $N=1$ we have $s(x)=0$, so the identity circuit suffices.
We assume $N\geq2$. Besides the input register $x$ and target register
$z$, we use one ancillary bit $a_{ij}$ for each pair
$1\leq i<j\leq N$. Each is initially zero. The bit $a_{ij}$ supplies
the negated value of input $i$ to the calculation of prefix $j$.
There are $T$ such bits.

We implement $C$ in five layers. Throughout the calculation,
$x_i$ and $z_j$ denote the original input values.

\textbf{Layer 1.}
We apply NOT to the first $N-1$ input bits. They now contain $1-x_i$.
The last input bit is unchanged because it does not occur in any
prefix.

\textbf{Layer 2.}
For each $i<N$ we apply one $\fanout_{N-i+1}$ gate with input bit $i$
as its control and $a_{i,i+1},\ldots,a_{iN}$ as its targets. This gives
\begin{equation*}
    a_{ij}=1-x_i
    \qquad(1\leq i<j\leq N).
\end{equation*}
The fanouts have distinct controls and disjoint sets of targets, so
they form one layer.

\textbf{Layer 3.}
For each $j=2,\ldots,N$ we apply $\toff_j$ with controls
$a_{1j},\ldots,a_{j-1,j}$ and target bit $j$ of $z$. The target becomes
\begin{equation*}
    z_j\oplus\bigwedge_{i=1}^{j-1}(1-x_i).
\end{equation*}
Each gate has its own controls and target. We can thus apply all
$N-1$ gates in one layer. The ancillary control bits remain unchanged.

\textbf{Layer 4.}
We repeat the fanouts from Layer 2 to return every $a_{ij}$ to zero.
In the same layer we apply NOT to target bits $2,\ldots,N$.
These operations can run together because the target register is
disjoint from the registers used by the fanouts. Each target now
contains
\begin{equation*}
    z_j\oplus\bigwedge_{i=1}^{j-1}(1-x_i)\oplus1
    =z_j\oplus s_j(x).
\end{equation*}
The first target bit is unchanged, as required by $s_1(x)=0$.

\textbf{Layer 5.}
We undo the NOT gates from Layer 1 to restore the input register.
All temporary copies are already zero, so the complete output is
$(x,z\oplus s(x),0^T)$.

The five layers use respectively $N-1$, $N-1$, $N-1$, $2(N-1)$
and $N-1$ elementary gates. Their sum is $6(N-1)$.
The two registers and the $T$ temporary copies give width $2N+T$.
Reversing the circuit gives $C^{-1}$ with the same depth and gate
count.

NOT, generalised Toffoli and fanout act on quantum states by
permuting the computational basis without introducing phases.
The calculation above gives
\begin{equation}
    C\bigl(\ket{x}\ket{z}\ket{0}^{\otimes T}\bigr)
    =\ket{x}\ket{z\oplus s(x)}\ket{0}^{\otimes T}
\end{equation}
on every basis input. By linearity this identity holds on arbitrary
superpositions, with all ancillary qubits returned to zero.
\end{proof}

The circuit uses $O(N)$ elementary gates and $O(N^2)$ bits or qubits.
The separate copies allow all prefix conjunctions to be computed in
one layer. Its total number of gate inputs is also $O(N^2)$.
For $N\geq2$ the two fanout layers contribute $2T+2(N-1)$ gate inputs,
the Toffoli layer contributes $T+N-1$, and the NOT gates contribute
$3(N-1)$. The total is $3T+6(N-1)$.

\section{Probabilistic classical circuits}
\label{sec:prob-classical-ckts}

We now extend the reversible classical model to allow probabilistic operations. Notably, our models are related to recent developments in classical computing with probabilistic bits~\cite{Kaiser2021pbits,Aadit21pbits,Chowdhury2025pbits}. We represent a distribution over bitstrings by a column vector and the action of a gate by a column stochastic matrix. We retain the conventions for gate locations and circuit layers from Section~\ref{circuit-conventions}. 

\begin{dfn}
    A classical probabilistic gate on $m$ bits is a $2^m\times 2^m$ column stochastic matrix $M$. Its entries are nonnegative real numbers, and each column sums to 1. Its action on a probability distribution, represented by a column vector $p$, is $p\mapsto Mp$. Parallel composition of gates $M_1,M_2$ is given by $M_1\otimes M_2$, and applying $M_1$ followed by $M_2$ gives $M_2M_1$.
\end{dfn}

\begin{eg}
    The single-bit erasure gate sets the bit to 0 for either input. Its matrix is
    \begin{equation}
        \left(\begin{array}{cc}
           1 & 1\\
           0 & 0
        \end{array}\right).
    \end{equation}
    The controlled erasure gate leaves the first bit unchanged and sets the second bit to 0 if the first bit is 1. In the basis order $00,01,10,11$, its matrix is
    \begin{equation}
        \left(\begin{array}{cccc}
            1 & 0 & 0 & 0\\
            0 & 1 & 0 & 0\\
            0 & 0 & 1 & 1\\
            0 & 0 & 0 & 0
        \end{array}\right).
    \end{equation}
\end{eg}

We build circuits from arbitrary single-bit stochastic gates, $\toff$ gates and $\fanout$ gates. We will use these elementary gates to prepare any probability distribution while returning all additional bits to zero.

\begin{dfn}
    Let $M$ be a classical probabilistic gate on $m$ bits, where $m\leq n$, and let $b_1,\ldots,b_m$ be distinct locations in $[n]$. A circuit $C$ on $n$ bits realises $M$ at these locations if it acts as $M$ on this register and returns every other bit to 0 whenever those other bits are initially 0. It suffices to check this on each definite input, since the action on probability distributions follows by linearity.
\end{dfn}

\begin{prop}
    The controlled erasure gate can be implemented in depth 3 with one ancillary bit initially set to 0.
\end{prop}
\begin{proof}
    Let $x_1,x_2$ be the input bits and $x_3=0$ the ancillary bit. We first compute $x_1x_2$ into $x_3$ using $\toff_3$. We then apply $\toff_2$ with control $x_3$ and target $x_2$, and erase $x_3$ with the single-bit erasure gate. These three layers give
    \begin{equation}
        \bigl(x_1,x_2\oplus(x_1x_2),0\bigr)
        =\bigl(x_1,x_2(1-x_1),0\bigr).
    \end{equation}
    We obtain the required action on each definite input, with the ancillary bit returned to zero. Linearity gives the action on an arbitrary probability distribution.
\end{proof}

\begin{rmk}
    To make the connection with quantum circuits explicit, we write $\ket{x}$ for the basis vector corresponding to $x\in\zo^n$. A probability vector has the form $\ket{p}=\sum_x p_x\ket{x}$, where $p_x\geq0$ and $\sum_xp_x=1$. Its coefficients are probabilities, rather than quantum amplitudes.
\end{rmk}

\begin{rmk}
    Permutation matrices are stochastic matrices, so this model contains the reversible classical circuits of the previous section.
\end{rmk}

\begin{dfn}
    Let $m\leq n$. A probabilistic classical gate $M$ on $n$ bits prepares a probability vector $\ket{p}$ on $m$ bits if
    \begin{equation}
        M\ket{0^n}=\ket{p}\ket{0^{n-m}}.
    \end{equation}
    A probabilistic classical circuit prepares $\ket{p}$ if its matrix prepares $\ket{p}$.
\end{dfn}

To prepare a distribution, we first choose bits independently and then erase every 1 after the first. By choosing the individual probabilities appropriately, we make the position of the first 1 follow the desired distribution, with the all-zero string accounting for one further outcome. We then convert this representation to an indicator vector and recover the corresponding bitstring using the inverse circuit from Lemma~\ref{permutation to indicator in constant depth}.

\begin{thm}
\label{Probability distribution vector in constant depth}
Any probability vector on $n\geq1$ bits can be prepared by a classical
probabilistic circuit of depth at most $29$, using $O(n2^n)$ elementary
gates and $O(4^n)$ bits. The gate set consists of arbitrary single-bit
stochastic gates, generalised Toffoli gates and fanout gates.
Every bit outside the output register is returned to zero.
\end{thm}

\begin{proof}
Let $N=2^n-1$ and write the target probability vector as
\begin{equation*}
    \ket{p}=\sum_{j=0}^N p_j\ket{\mathrm{bin}(j)},
\end{equation*}
where $\mathrm{bin}(j)$ has $n$ bits, $p_j\geq0$ and
$\sum_{j=0}^N p_j=1$.

We may assume $p_0>0$. To obtain this condition, choose a bitstring
of positive probability and relabel every bitstring by its XOR with
that string. After preparing the relabelled distribution, we undo
the relabelling in one layer of NOT gates.

Set
\begin{equation*}
    q_i=\frac{p_i}{p_0+p_i+\cdots+p_N}
    \qquad(1\leq i\leq N).
\end{equation*}
The condition $p_0>0$ gives $0\leq q_i<1$. Since
\begin{equation*}
    1-q_i
    =\frac{p_0+p_{i+1}+\cdots+p_N}
           {p_0+p_i+\cdots+p_N},
\end{equation*}
successive factors cancel to give
\begin{equation*}
    p_j=q_j\prod_{i<j}(1-q_i)
    \quad(1\leq j\leq N),
    \qquad
    p_0=\prod_{i=1}^N(1-q_i).
\end{equation*}
An empty sum is zero and an empty product is one.
These identities give the probability that the first 1 occurs at
position $j$ and the probability that all bits are zero.

\textbf{Step 1.}
We prepare $N$ independent bits, with bit $i$ equal to 1 with
probability $q_i$. This uses $N$ single-bit stochastic gates in one
layer. The first 1 is in position $j$ with probability $p_j$, and all
bits are zero with probability $p_0$.

\textbf{Step 2.}
We introduce a second register of $N$ zero bits and apply the circuit
$C$ from Lemma~\ref{prefix-values-linear-gates}.
For each input string $x$, it writes $s(x)$ into the second register
and returns its $T=N(N-1)/2$ working bits to zero.
This takes at most five layers and $6(N-1)$ gates.

\textbf{Step 3.}
For each $i$ we set bit $i$ of the first register to zero whenever
bit $i$ of the second register is 1. We implement this operation
using one additional zero bit. A Toffoli first writes the AND of
the control and target into that bit. A CNOT from the additional
bit to the target then sets the target to zero when the control
is 1. Finally, a single-bit erasure gate returns the additional
bit to zero.

The control, target and ancillary bits are separate for each
position. All $N$ operations run in three layers using $3N$
elementary gates.

Write
\begin{equation*}
    y^0=0^N,\qquad
    y^j=0^{j-1}10^{N-j}
    \quad(1\leq j\leq N).
\end{equation*}
Every string whose first 1 is at position $j$ becomes $y^j$.
Its prefix values remain zero through position $j$ and one
afterwards. The joint probability vector is now
\begin{equation*}
    \sum_{j=0}^N p_j\ket{y^j}\ket{s(y^j)}.
\end{equation*}
All ancillary bits used for erasure have returned to zero.

\textbf{Step 4.}
We apply $C^{-1}$ to clear the second register. The identity
\begin{equation*}
    C(y^j,0^N,0^T)=(y^j,s(y^j),0^T)
\end{equation*}
shows that the inverse returns the second register and all working
bits to zero for every $j$. This takes at most five layers and
$6(N-1)$ gates.

We next prepend a zero bit to distinguish the all-zero outcome.
We apply $\revor_{N+1}$ with the new bit as its target and the
$N$ bits of the first register as its controls, then apply NOT
to the target. The new bit is 1 precisely when the first register
is all zero.

A reversible OR can be implemented in three layers by negating
its controls, applying a generalised Toffoli, and then negating
the controls and target. Including the final NOT gives at most
four layers and $O(N)$ gates. The resulting probability vector is
\begin{equation*}
    \sum_{j=0}^N p_j\ket{0^j10^{N-j}}.
\end{equation*}

We now apply the inverse indicator circuit from
Lemma~\ref{permutation to indicator in constant depth} in
lexicographical order. We place the $N+1$ indicator bits at its
indicator output locations and initialise all other bits to zero.
The inverse maps the indicator in position $j+1$ to
$\mathrm{bin}(j)$ and clears every other bit. This takes at most
ten layers and $O(n2^n)$ gates. We finish by undoing any initial
XOR relabelling.

The total depth is at most
\begin{equation*}
    1+5+3+5+4+10+1=29.
\end{equation*}
Computing and clearing the prefixes uses at most $12(N-1)$ gates.
The initial preparation, erasures, additional indicator and final
NOT gates use $O(N+n)$ gates. Decoding uses $O(n2^n)$ gates, giving
$O(n2^n)$ in total.

The prefix circuit uses $O(N^2)$ bits and all remaining stages
together require $O(N+n2^n)$ bits. Even with separate working
registers for these stages, the width is $O(4^n)$.
Every bit outside the decoded output returns to zero.
\end{proof}

\section{Quantum circuits}
\label{quantum circuits}

We begin by explaining how to use the reversible constructions of Section~2 in quantum circuits. Rosenthal \cite{rosenthal2020boundsqac0complexityapproximating} proved that $\fanout$ can be approximated by exponential size $\qaczero$ circuits. The result of Grier, Morris and Wu \cite{grier2026mathsfqac0containsmathsftc0with}, together with the parity construction and its conversion to fanout in \cite{rosenthal2020boundsqac0complexityapproximating}, gives an exact implementation with clean ancillary qubits. Its depth is bounded by a constant, while its size and number of ancillary qubits are exponential in the number of qubits on which the fanout acts. We may thus include $\fanout$ without changing what $\qaczero$ can implement when size is unrestricted. In particular, by applying Proposition \ref{Any permutation in constant depth}, we obtain an exact $\qaczero$ implementation of any permutation of computational basis states.

We give the resource bounds below for $\mathsf{QAC^0_f}$, where fanout is included in the gate set. If a fanout acts on exponentially many qubits, replacing it by the $\qaczero$ implementation above may increase our size bound to a doubly exponential function of the original input length. This is an upper bound obtained by replacing fanout in our constructions.

We say that a circuit $C$ prepares an $n$-qubit pure state $\ket{\psi}$ if, for some number $c$ of ancillary qubits,
\begin{equation}
    C\ket{0}^{\otimes(n+c)}=\ket{\psi}\ket{0}^{\otimes c}.
\end{equation}
As in Section~\ref{circuit-conventions}, we require every ancillary qubit to return to zero. We can then run the preparation circuit backwards to map $\ket{\psi}\ket{0}^{\otimes c}$ to the all-zero state.

We first implement an arbitrary diagonal unitary by computing an indicator for each computational basis string. On a basis input, exactly one indicator is 1. We apply the corresponding phase gate to every indicator in parallel, so only the phase assigned to the input appears. We then reverse the indicator calculation to clear all additional qubits. By linearity, this gives the required action on every superposition.

\begin{prop}
\label{diagonal phase unitaries in constant depth}
    Any diagonal phase gate on $n$ qubits can be implemented by a $\mathsf{QAC^0_f}$ circuit of depth seven, using $O(n2^n)$ qubits and gates.
\end{prop}
\begin{proof}
Let $N=2^n$ and let $x^1,\ldots,x^N$ list the elements of $\zo^n$ in lexicographical order. We write
\begin{equation}
    U=\sum_{i=1}^N\alpha_i\ket{x^i}\bra{x^i},
    \qquad |\alpha_i|=1.
\end{equation}
We arrange the qubits in an $N\times n$ array $A$ and an $N$-qubit column $B$. The first row of $A$ contains the input and all other qubits start in zero. The total number of qubits is $(n+1)N$. We describe the circuit on a computational basis input $\ket{x}$ with $x=x^s$. Figure~\ref{diagonal construction stages} illustrates the construction for $n=3$ and $x=101$.

In the first layer we apply a $\fanout_N$ gate to each column of $A$, using its first qubit as the control. These gates act on disjoint columns and place $x$ in every row. In the second layer we apply NOT to each entry $A_{ij}$ for which $(x^i)_j=0$. Row $i$ then consists entirely of ones exactly when $x=x^i$.

In the third layer we apply one $\toff_{n+1}$ gate to each row of $A$, with that row as its controls and $B_i$ as its target. These gates act on disjoint qubits. They prepare the indicator register
\begin{equation}
    \ket{\mathbb{1}_{x^1}(x)}\cdots\ket{\mathbb{1}_{x^N}(x)}.
\end{equation}
Its only nonzero bit is $B_s$.

In the fourth layer we apply
\begin{equation}
    Z(\alpha_i)=
    \begin{pmatrix}
        1&0\\
        0&\alpha_i
    \end{pmatrix}
\end{equation}
to each $B_i$. The resulting phase is
\begin{equation}
    \prod_{i=1}^N\alpha_i^{\mathbb{1}_{x^i}(x)}=\alpha_s.
\end{equation}

We then reverse the Toffoli, NOT and fanout layers. These three layers restore every ancillary qubit to zero and give
\begin{equation}
    \ket{x^s}\ket{0}^{\otimes((n+1)N-n)}
    \longmapsto
    \alpha_s\ket{x^s}\ket{0}^{\otimes((n+1)N-n)}.
\end{equation}
Linearity gives the required action of $U$ on an arbitrary input. The circuit has seven layers, $(n+1)N$ qubits and $O(nN)$ gates.
\end{proof}

\begin{figure}[!ht]
    \centering
\begingroup
\definecolor{diagblue}{RGB}{25,78,110}
\begin{tikzpicture}[
    font=\normalfont\footnotesize,
    every node/.style={align=center},
    stagearrow/.style={->,>=latex,line width=.75pt},
    operation/.style={draw=diagblue,line width=.55pt},
    gridline/.style={draw=black!60,line width=.35pt}
]
\newcommand{\diagonalpanel}[7]{%
    \begin{scope}[shift={#1}]
        \node[font=\normalfont\small\bfseries] at (2.05,.70) {#2};
        \node at (.40,.19) {$x^i$};
        \node at (1.67,.19) {$A$};
        \node at (3.09,.19) {$B$};
        \ifnum#5>0
            \pgfmathsetmacro{\rowtop}{-(#5-1)*.30}
            \fill[diagblue!12] (.95,\rowtop) rectangle (2.39,\rowtop-.30);
            \fill[diagblue!12] (2.85,\rowtop) rectangle (3.33,\rowtop-.30);
        \fi
        \foreach \j in {0,1,2,3}{
            \draw[gridline] (.95+\j*.48,0) -- (.95+\j*.48,-2.40);
        }
        \foreach \r in {0,...,8}{
            \draw[gridline] (.95,-\r*.30) -- (2.39,-\r*.30);
            \draw[gridline] (2.85,-\r*.30) -- (3.33,-\r*.30);
        }
        \draw[gridline] (2.85,0) -- (2.85,-2.40);
        \draw[gridline] (3.33,0) -- (3.33,-2.40);
        \foreach \candidate [count=\r from 1] in {000,001,010,011,100,101,110,111}{
            \node[font=\normalfont\footnotesize] at (.40,-\r*.30+.15) {$\candidate$};
        }
        \foreach \value [count=\k from 0] in {#3}{
            \pgfmathtruncatemacro{\r}{floor(\k/3)}
            \pgfmathtruncatemacro{\c}{mod(\k,3)}
            \node at (1.19+\c*.48,-\r*.30-.15) {$\value$};
        }
        \foreach \value [count=\r from 1] in {#4}{
            \node at (3.09,-\r*.30+.15) {$\value$};
        }
        \node[anchor=north,text width=4.6cm,font=\normalfont\fontsize{8}{10}\selectfont] at (2.05,-2.60) {\hyphenpenalty=10000 #6};
        #7
    \end{scope}
}
\diagonalpanel{(0,0)}{(a) Initial registers}
    {1,0,1,0,0,0,0,0,0,0,0,0,0,0,0,0,0,0,0,0,0,0,0,0}
    {0,0,0,0,0,0,0,0}{0}
    {$x=101=x^6$ is in the first row.\\All other qubits are zero.}{}
\diagonalpanel{(5.05,0)}{(b) Fanout down columns}
    {1,0,1,1,0,1,1,0,1,1,0,1,1,0,1,1,0,1,1,0,1,1,0,1}
    {0,0,0,0,0,0,0,0}{0}
    {Copy the top qubit down each column.\\One layer.}
    {\foreach \c in {0,1,2}{
        \draw[operation,->,>=latex] (1.36+\c*.48,-.33) -- (1.36+\c*.48,-2.30);
    }}
\diagonalpanel{(10.10,0)}{(c) Apply NOT gates}
    {0,1,0,0,1,1,0,0,0,0,0,1,1,1,0,1,1,1,1,0,0,1,0,1}
    {0,0,0,0,0,0,0,0}{6}
    {Flip $A_{ij}$ when $(x^i)_j=0$.\\Only row $6$ is all ones. One layer.}{}
\diagonalpanel{(10.10,-5.05)}{(d) Write the indicators}
    {0,1,0,0,1,1,0,0,0,0,0,1,1,1,0,1,1,1,1,0,0,1,0,1}
    {0,0,0,0,0,1,0,0}{6}
    {Row $i$ controls a Toffoli gate\\with target $B_i$. One layer.}
    {\foreach \r in {1,...,8}{
        \draw[operation,->,>=latex] (2.42,-\r*.30+.15) -- (2.80,-\r*.30+.15);
    }}
\diagonalpanel{(5.05,-5.05)}{(e) Apply phases}
    {0,1,0,0,1,1,0,0,0,0,0,1,1,1,0,1,1,1,1,0,0,1,0,1}
    {0,0,0,0,0,1,0,0}{6}
    {Apply $Z(\alpha_i)$ to each $B_i$.\\The state gains phase $\alpha_6$. One layer.}
    {\node[text=diagblue] at (3.84,-1.65) {$\alpha_6$};
     \draw[operation] (3.36,-1.65) -- (3.53,-1.65);}
\diagonalpanel{(0,-5.05)}{(f) Restore the registers}
    {1,0,1,0,0,0,0,0,0,0,0,0,0,0,0,0,0,0,0,0,0,0,0,0}
    {0,0,0,0,0,0,0,0}{0}
    {Undo Toffoli, NOT and fanout.\\Three layers leave $\alpha_6\ket{101}\ket{0}^{\otimes29}$.}{}
\draw[stagearrow] (4.37,-1.20) -- (4.83,-1.20);
\draw[stagearrow] (9.42,-1.20) -- (9.88,-1.20);
\draw[stagearrow] (12.15,-3.50) -- (12.15,-3.96);
\draw[stagearrow] (9.88,-6.25) -- (9.42,-6.25);
\draw[stagearrow] (4.83,-6.25) -- (4.37,-6.25);
\end{tikzpicture}
\endgroup
    \caption{The diagonal construction for $n=3$ and $x=101=x^6$. Each cell gives the value of one qubit. The strings $x^i$ label the rows and are not additional registers. Follow the panels clockwise from (a) to (f). Column fanout copies the input, and NOT gates make row $6$ the unique row consisting entirely of ones. The row Toffoli gates compute the indicator for this row, and the phase gates apply $\alpha_6$. Panels (b)--(e) each add one circuit layer. Panel (f) reverses the first three layers, giving seven layers in total. The grid describes the registers, rather than their physical arrangement. The action on arbitrary inputs follows by linearity.}
    \label{diagonal construction stages}
\end{figure}
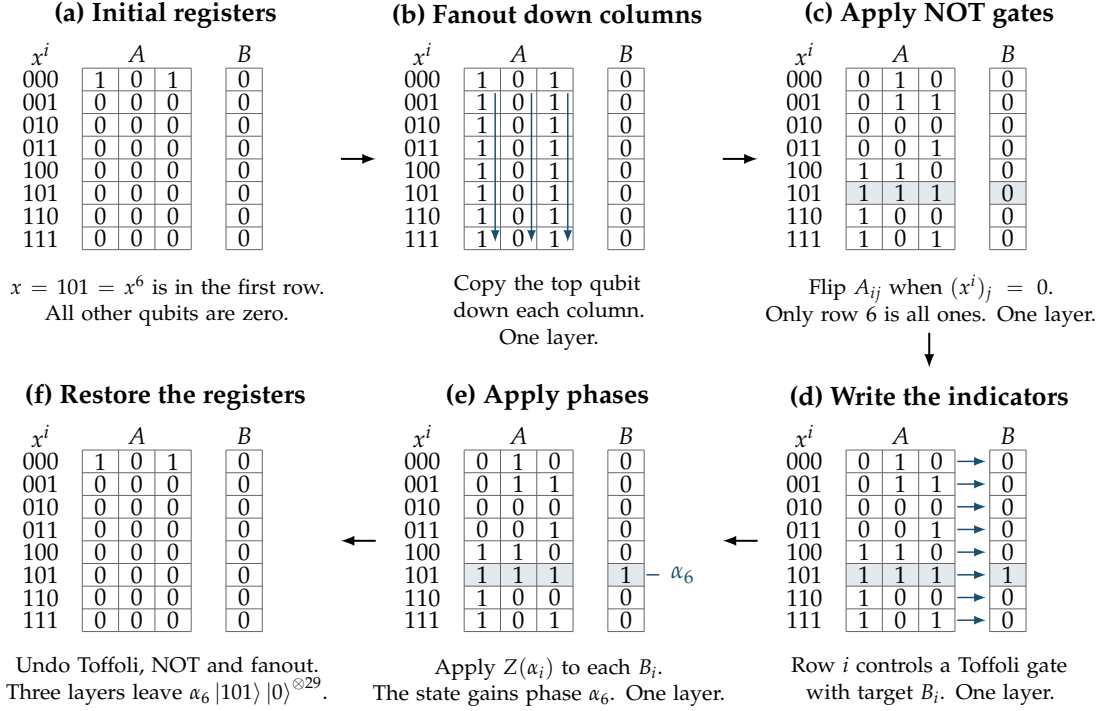

We next adapt the classical probability construction from Section~3 to prepare arbitrary pure states. Exact preparation in $\mathsf{QAC^0_f}$ with clean ancillary qubits is already known \cite{rosenthal2026querydepth}, using $2^n$ qubits up to a factor polynomial in $n$. Gretta, Gupta and Joshi~\cite[Theorem~1.1]{gretta2026logarithmic} give a constant-depth construction using only single-qubit and generalised Toffoli gates. It uses $2^{O(n)}$ qubits and gates and returns all ancillary qubits to zero. Our construction uses $O(4^n)$ qubits and gates with fanout included in the gate set. We give it to explain the connection with classical probability preparation and the role of reversibility in the quantum construction.

\begin{thm}
\label{Any state in constant depth}
Let $\ket{\psi}$ be any pure state on $n\geq1$ qubits.
Then a $\mathsf{QAC^0_f}$ circuit of depth at most $37$ prepares
$\ket{\psi}$ using $O(n2^n)$ elementary gates and $O(4^n)$ qubits.
Every ancillary qubit returns to $\ket{0}$.
\end{thm}

\begin{proof}
We first consider
\begin{equation*}
    \ket{\psi}
    =\sum_{j=0}^N\sqrt{p_j}\ket{\mathrm{bin}(j)},
    \qquad N=2^n-1,
\end{equation*}
with $p_0>0$, $p_j\geq0$ and $\sum_{j=0}^N p_j=1$.
A general target reduces to this case by taking the absolute values
of its coefficients and relabelling the basis strings by a fixed
XOR so that the coefficient at $0^n$ is nonzero. After preparation
we undo the relabelling and restore the phases.

Define
\begin{equation*}
    q_i=\frac{p_i}{p_0+p_i+\cdots+p_N}
    \qquad(1\leq i\leq N).
\end{equation*}
The probability identities in the classical proof determine the
amplitudes in the construction below.

\textbf{Step 1.}
We apply the rotations
\begin{equation*}
    R_i=
    \begin{pmatrix}
        \sqrt{1-q_i}&-\sqrt{q_i}\\
        \sqrt{q_i}&\sqrt{1-q_i}
    \end{pmatrix}
\end{equation*}
to $N$ zero qubits in parallel. Grouping the resulting basis strings
by the position of their first 1 gives
\begin{equation}
    \bigotimes_{i=1}^N R_i\ket{0}
    =\sqrt{p_0}\ket{0^N}
+\sum_{j=1}^N\sqrt{p_j}\ket{0^{j-1}1}
       \bigotimes_{i=j+1}^N R_i\ket{0}.
\end{equation}
The coefficient of the term whose first 1 is at position $j$ is
\begin{equation*}
    \sqrt{q_j}\prod_{i<j}\sqrt{1-q_i}=\sqrt{p_j}.
\end{equation*}
The remaining qubits in that term form a normalised product state.
We omit tensor products over no qubits.

\textbf{Step 2.}
We compute the prefix values into a second register initially in
$\ket{0^N}$ using Lemma~\ref{prefix-values-linear-gates}.
This takes at most five layers and $6(N-1)$ gates and returns
the $T=N(N-1)/2$ working qubits to zero.

If the first 1 is at position $j$, every basis string in the
remaining tensor product has the same prefix values
$0^j1^{N-j}$. The joint state is
\begin{equation}
    \sqrt{p_0}\ket{0^N}\ket{0^N}\\
    +\sum_{j=1}^N\sqrt{p_j}
       \left(
           \ket{0^{j-1}1}
\bigotimes_{i=j+1}^N R_i\ket{0}
       \right)
       \ket{0^j1^{N-j}}.
\end{equation}
Within each term the prefix values are independent of the later
bits. Computing them thus preserves the product state of the
later qubits.

\textbf{Step 3.}
For every $i$ we apply $R_i^\dagger$ to qubit $i$ of the first
register when qubit $i$ of the second register is 1.
In the term whose first 1 is at position $j$, these operations
return precisely the qubits with $i>j$ to zero. They preserve
the earlier zeros, the first 1 and the coefficient $\sqrt{p_j}$.
The all-zero term is also unchanged.

Writing
\begin{equation*}
    y^0=0^N,\qquad
    y^j=0^{j-1}10^{N-j}
    \quad(1\leq j\leq N),
\end{equation*}
we obtain
\begin{equation*}
    \sum_{j=0}^N\sqrt{p_j}\ket{y^j}\ket{s(y^j)}.
\end{equation*}

We implement each conditional inverse rotation using two
single-qubit rotations and two CNOT gates. If $R_i$ is the real
rotation by angle $\theta_i$, we apply, in order, the rotation
by $-\theta_i/2$, a CNOT, the rotation by $\theta_i/2$ and a
second CNOT. Both CNOT gates use the corresponding prefix qubit
as their control and qubit $i$ of the first register as their
target. For control 0 the rotations cancel. For control 1,
conjugation by $X$ reverses the rotation angle, so the combined
action is $R_i^\dagger$.

The control and target pairs are disjoint for different $i$.
This step takes four layers and $4N$ elementary gates.

\textbf{Step 4.}
We apply $C^{-1}$ to the two registers and the zero working qubits.
On each term it acts as
\begin{equation*}
    \ket{y^j}\ket{s(y^j)}\ket{0}^{\otimes T}
    \longmapsto
    \ket{y^j}\ket{0^N}\ket{0}^{\otimes T}.
\end{equation*}
Returning the qubits after the first 1 to zero does not change
the prefix values. The inverse thus clears the second register
while preserving every coefficient in the superposition.
This takes at most five layers and $6(N-1)$ gates.

\textbf{Step 5.}
We prepend a zero qubit and compute the indicator for the all-zero
term. We apply $\revor_{N+1}$ with the new qubit as its target and
the $N$ qubits of the first register as its controls, followed by
NOT on the target. This takes at most four layers and $O(N)$ gates
and gives
\begin{equation*}
    \sum_{j=0}^N\sqrt{p_j}\ket{0^j10^{N-j}}.
\end{equation*}

\textbf{Step 6.}
We apply the inverse indicator circuit from
Lemma~\ref{permutation to indicator in constant depth} in
lexicographical order. We place the indicators at the corresponding
output locations of that circuit and initialise its other registers
to zero. In at most ten layers and $O(n2^n)$ gates, the inverse
produces
\begin{equation*}
    \sum_{j=0}^N\sqrt{p_j}\ket{\mathrm{bin}(j)}
\end{equation*}
and returns every other qubit to zero.

For a general target we undo the XOR relabelling in one layer
and restore the phases using
Proposition~\ref{diagonal phase unitaries in constant depth}.
The diagonal circuit acts on the $n$ decoded output qubits.
It uses seven layers and $O(n2^n)$ gates and additional qubits,
all of which return to zero. Where a target coefficient is zero,
we may choose its diagonal phase arbitrarily.

Including the final relabelling and phase restoration, the depth
is at most
\begin{equation*}
    1+5+4+5+4+10+1+7=37.
\end{equation*}
Computing and clearing the prefixes uses at most $12(N-1)$ gates.
The rotations and additional indicator use $O(N)$ gates.
Decoding and phase restoration each use $O(n2^n)$ gates, giving
$O(n2^n)$ in total.

The prefix circuit requires $O(N^2)$ qubits and all other stages
together use $O(N+n2^n)$. The total width is $O(4^n)$.
Every ancillary qubit returns to zero exactly.
\end{proof}

\section{Unitaries}

We have shown that unbounded size $\qaczero$ circuits can compute arbitrary Boolean functions reversibly, implement arbitrary permutations of computational basis states, and prepare arbitrary pure quantum states. We now ask whether they can implement every unitary on an arbitrary input. State preparation alone does not answer this question, since its required action is specified only on the zero input. We give several equivalent formulations of the unitary implementation problem, although we do not resolve it.

Throughout this section, we write $N=2^n$ and $\{\ket{\psi_1},\ldots,\ket{\psi_N}\}$ is an orthonormal basis of quantum states on $n$ qubits. 

\subsection{Cloning an orthonormal basis}

We first reduce the implementation of a unitary to copying the vectors of a basis in which it is diagonal. For the computational basis, we can make the required copies by applying one fanout gate to each of the $n$ input qubits. For a general orthonormal basis, the copying map still extends to a unitary, but we do not know whether it has a circuit of constant depth. The reduction assumes such an implementation, with a depth bound independent of the basis and of $n$.

\begin{prop}
\label{copying basis vectors}
Let $U$ be a unitary with decomposition
\[
U=\sum_{i=1}^{N}\alpha_i\ket{\psi_i}\bra{\psi_i},
\]
where $\{\ket{\psi_1},\ldots,\ket{\psi_N}\}$ is an orthonormal basis. Implementing $U$ by unbounded size $\qaczero$ circuits reduces to implementing, in the same circuit model, a unitary $V$ satisfying
\[
V\ket{\psi_i}\ket{0}^{\otimes K}
 =\ket{\psi_i}^{\otimes N}\ket{\phi_i}
 \qquad\text{for every }i\in[N].
\]
The normalised states $\ket{\phi_i}$ may depend on $i$.
\end{prop}
\begin{proof}
We apply $V$ and then act on block $j$ by a unitary that gives $\ket{\psi_j}$ the phase $\alpha_j$ and leaves the other basis vectors unchanged. These operations use separate blocks and can be performed in parallel. On input $\ket{\psi_i}$, only block $i$ acquires a phase. We then apply $V^\dagger$ to restore all additional qubits to zero, leaving the phase $\alpha_i$ on the input.

To implement the operation on block $j$, we use the constant depth preparation circuit $U_j$ supplied by Proposition \ref{Any state in constant depth}, so that
\begin{equation}
U_j\ket{0}^{\otimes(n+c)}
 =\ket{\psi_j}\ket{0}^{\otimes c}.
\end{equation}
We choose the same $c$ for every $j$ by adding unused zero qubits where necessary, and define
\begin{equation}
D_j=\alpha_j\ket{0}^{\otimes(n+c)}\bra{0}^{\otimes(n+c)}
 +\sum_{x\in\zo^{n+c}\setminus\{0^{n+c}\}}\ket{x}\bra{x}.
\end{equation}
The operator $D_j$ applies a phase only to the all-zero state. We implement it in depth five using one additional qubit initially in zero. To do this, we negate the $n+c$ inputs, compute their conjunction into the additional qubit with a Toffoli gate, apply $\operatorname{diag}(1,\alpha_j)$ to that qubit, and reverse the Toffoli and NOT gates. This returns the additional qubit to zero.

We include all working qubits of the circuit implementing $V$ in $K$, appending zero factors to $\ket{\phi_i}$ if necessary, and set $M=K+N(c+1)$. The arrays below show the $N$ basis registers as rows. Each row also has $c$ preparation qubits and one temporary phase qubit, all initially zero. The preparation qubits appear when they are used. The phase qubits are omitted, since they return to zero after each $D_j$. We retain $\ket{\phi_i}$ unchanged throughout these row operations so that we can subsequently reverse $V$.

Initially, the first basis register contains $\ket{\psi_i}$ and the others contain zeros:
\begin{equation}
N\text{ blocks}\left\{
\overbrace{\begin{array}{|ccc|}
---&\ket{\psi_i}&---\\
\ket{0}&\ldots&\ket{0}\\
\ket{0}&\ldots&\ket{0}\\
\vdots&\ddots&\vdots\\
\ket{0}&\ldots&\ket{0}
\end{array}}^{n\text{ qubits in each block}}
\right.
\end{equation}
We apply $V$ to the first $n+K$ qubits, obtaining the complete state
\[
\ket{\psi_i}^{\otimes N}\ket{\phi_i}\ket{0}^{\otimes N(c+1)}.
\]
The basis and preparation registers are
\begin{equation}
\underbrace{\begin{array}{|ccc|}
---&\ket{\psi_i}&---\\
---&\ket{\psi_i}&---\\
---&\ket{\psi_i}&---\\
&\vdots&\\
---&\ket{\psi_i}&---
\end{array}}_{n\text{ qubits per row}}
\qquad
\underbrace{\begin{array}{|ccc|}
\ket{0}&\ldots&\ket{0}\\
\ket{0}&\ldots&\ket{0}\\
\ket{0}&\ldots&\ket{0}\\
\vdots&\ddots&\vdots\\
\ket{0}&\ldots&\ket{0}
\end{array}}_{c\text{ qubits per row}}.
\end{equation}
If $c=0$, the second array is absent.

We next apply $U_j^\dagger$ to row $j$ and its preparation qubits, for all $j$ in parallel. The identity
\begin{equation}
\bra{0}^{\otimes(n+c)}
U_j^\dagger\bigl(\ket{\psi_i}\ket{0}^{\otimes c}\bigr)
 =\braket{\psi_j|\psi_i}=\delta_{ij},
\end{equation}
shows that row $i$ becomes the all-zero state, while the state in every other row is orthogonal to the all-zero state:
\begin{equation}
\begin{array}{|ccc|}
---&U_1^\dagger\bigl(\ket{\psi_i}\ket{0}^{\otimes c}\bigr)&---\\
---&U_2^\dagger\bigl(\ket{\psi_i}\ket{0}^{\otimes c}\bigr)&---\\
&\vdots&\\
---&\ket{0}^{\otimes(n+c)}&---\\
&\vdots&\\
---&U_N^\dagger\bigl(\ket{\psi_i}\ket{0}^{\otimes c}\bigr)&---
\end{array}.
\end{equation}
We can now apply all the $D_j$ in parallel, giving
\begin{equation}
\alpha_i\,
\begin{array}{|ccc|}
---&U_1^\dagger\bigl(\ket{\psi_i}\ket{0}^{\otimes c}\bigr)&---\\
---&U_2^\dagger\bigl(\ket{\psi_i}\ket{0}^{\otimes c}\bigr)&---\\
&\vdots&\\
---&\ket{0}^{\otimes(n+c)}&---\\
&\vdots&\\
---&U_N^\dagger\bigl(\ket{\psi_i}\ket{0}^{\otimes c}\bigr)&---
\end{array}.
\end{equation}
Only row $i$ acquires a phase, since $D_j$ acts as the identity on the orthogonal complement of the all-zero state. The other rows and the retained state $\ket{\phi_i}$ are unchanged.

We apply all $U_j$ in parallel to restore the copies and their preparation qubits, and then apply $V^\dagger$. The basis registers return to
\begin{equation}
\alpha_i\,
\begin{array}{|ccc|}
---&\ket{\psi_i}&---\\
\ket{0}&\ldots&\ket{0}\\
\ket{0}&\ldots&\ket{0}\\
\vdots&\ddots&\vdots\\
\ket{0}&\ldots&\ket{0}
\end{array},
\end{equation}
and every additional qubit is zero. Thus the complete action is
\[
\ket{\psi_i}\ket{0}^{\otimes M}
 \longmapsto
\alpha_i\ket{\psi_i}\ket{0}^{\otimes M}.
\]
By linearity, this is the required action of $U$ on an arbitrary input. The circuits $U_j$ have a common constant depth bound, and their parallel execution does not increase it. Together with the assumed bound for $V$, this gives a constant depth bound for the complete circuit.
\end{proof}

Note that we do not need the rows to contain identical copies. The proof only uses the fact that, on input $\ket{\psi_i}$, row $i$ contains $\ket{\psi_i}$ and every other row $j$ is orthogonal to $\ket{\psi_j}$. 

\subsection{Permuting an orthonormal basis}

We next show that this property can also be obtained by arranging one copy of each basis vector in a suitable order. This means that the ability to map basis states to such permutations implies the ability to realise an arbitrary unitary.

\begin{prop}
\label{permuted basis vectors}
Let $n\geq2$. For each $i\in[N]$, let $\sigma_i$ be a permutation of $[N]$ whose only fixed point is $i$, and let
\[
\ket{\eta_i}
 =\ket{\psi_{\sigma_i(1)}}\ket{\psi_{\sigma_i(2)}}\cdots
  \ket{\psi_{\sigma_i(N)}}.
\]
Suppose that unbounded size $\qaczero$ circuits can implement a unitary $V$ satisfying
\[
V\ket{\psi_i}\ket{0}^{\otimes K}
 =\ket{\eta_i}\ket{\phi_i}
 \qquad\text{for every }i\in[N],
\]
where the states $\ket{\phi_i}$ are normalised. They can then implement every unitary diagonal in the basis $\{\ket{\psi_1},\ldots,\ket{\psi_N}\}$.
\end{prop}
\begin{proof}
We obtain such a permutation by fixing $i$ and cycling the other $N-1$ indices. The assumption $n\geq2$ ensures that there are at least two other indices. When $N=2$, no permutation has exactly one fixed point, but the corresponding case $n=1$ needs no reduction because every single-qubit unitary is an elementary gate.

We write the target as $U=\sum_i\alpha_i\ket{\psi_i}\bra{\psi_i}$ and apply $V$. Retaining $\ket{\phi_i}$ unchanged, we act on the basis registers, which now contain
\begin{equation}
\begin{array}{|ccc|}
---&\ket{\psi_{\sigma_i(1)}}&---\\
---&\ket{\psi_{\sigma_i(2)}}&---\\
&\vdots&\\
---&\ket{\psi_i}&---\\
&\vdots&\\
---&\ket{\psi_{\sigma_i(N)}}&---
\end{array}.
\end{equation}
We use the preparation circuits $U_j$ and phase operators $D_j$ from Proposition \ref{copying basis vectors}, with $c$ preparation qubits and one phase qubit for each row. In row $j$, we have
\[
\bra{0}^{\otimes(n+c)}
 U_j^\dagger\bigl(\ket{\psi_{\sigma_i(j)}}\ket{0}^{\otimes c}\bigr)
 =\braket{\psi_j|\psi_{\sigma_i(j)}}.
\]
The inner product equals one when $j=i$ and zero otherwise, since $i$ is the only fixed point of $\sigma_i$. We apply the operations $U_jD_jU_j^\dagger$ in parallel. They multiply the entire state by $\alpha_i$ and restore their additional qubits to zero. Applying $V^\dagger$ then gives
\[
\alpha_i\ket{\psi_i}\ket{0}^{\otimes M},
\qquad M=K+N(c+1).
\]
As in the preceding proof, we include all working qubits of the implementation of $V$ in $K$. Applying $V^\dagger$ also reverses the preparation of the retained state $\ket{\phi_i}$. The required action follows by linearity, and the depth bound is the same as in Proposition \ref{copying basis vectors}.
\end{proof}

\subsection{Decoding basis labels}
We can produce the arrangements in Proposition~\ref{permuted basis vectors} if we can obtain a computational basis label identifying the input basis vector. We prepare all basis vectors in separate registers and use the label to select their order, retaining the label and any remaining qubits so that the operation can be reversed.

\begin{prop}
\label{basis labels}
Let $n\geq2$, and let $\sigma_i$ and $\ket{\eta_i}$ be as in Proposition \ref{permuted basis vectors}. Implementing
\[
\ket{\psi_i}\ket{0}^{\otimes M}
 \longmapsto \ket{\eta_i}\ket{\phi_i}
 \qquad\text{for every }i\in[N]
\]
in unbounded size $\qaczero$ reduces to implementing a unitary $V$ satisfying
\[
V\ket{\psi_i}\ket{0}^{\otimes K}
 =\ket{i}\ket{junk_i}.
\]
Here $\ket{i}$ denotes the $i$th computational basis vector on $n$ qubits, in the chosen ordering, and the states $\ket{junk_i}$ are normalised.
\end{prop}
\begin{proof}
We prepare $\ket{\psi_1},\ldots,\ket{\psi_N}$ in separate registers in parallel, using Proposition \ref{Any state in constant depth}. Their preparation qubits return to zero. We then apply $V$ to the unknown input and its $K$ additional qubits, retaining both $\ket{i}$ and $\ket{junk_i}$.

For each label $i$, we permute the prepared registers into the order specified by $\sigma_i$. The operation that retains the label and applies the corresponding register permutation is itself a permutation of computational basis states, so Proposition \ref{Any permutation in constant depth} gives its constant depth implementation. The output registers then contain $\ket{\eta_i}$, and the retained label, $\ket{junk_i}$ and unused zero qubits form $\ket{\phi_i}$. This gives the required map, with the label retained as part of the output.

This argument uses the absence of a size restriction: applying the permutation construction to exponentially many qubits need not give a circuit of exponential size in the original $n$.
\end{proof}

\begin{rmk}
We can also use a basis label to implement any unitary diagonal in that basis directly, by applying $V$, the desired diagonal phase on the label register, and then $V^\dagger$. The state $\ket{junk_i}$ remains unchanged between $V$ and $V^\dagger$.
\end{rmk}

\begin{cor}
\label{labels imply unitaries}
Let $U$ be a unitary with decomposition
\[
U=\sum_{i=1}^{N}\alpha_i\ket{\psi_i}\bra{\psi_i}
\]
for an orthonormal basis $\{\ket{\psi_1},\ldots,\ket{\psi_N}\}$. Implementing $U$ by unbounded size $\qaczero$ circuits reduces to implementing a unitary $V$ such that
\[
V\ket{\psi_i}\ket{0}^{\otimes K}
 =\ket{i}\ket{junk_i}
 \qquad\text{for every }i\in[N].
\]
\end{cor}
\begin{proof}
We apply the preceding remark with the diagonal unitary
$\sum_i\alpha_i\ket{i}\bra{i}$, which is available by Proposition \ref{diagonal phase unitaries in constant depth}. The resulting map is
$\ket{\psi_i}\ket{0}^{\otimes K}\mapsto
\alpha_i\ket{\psi_i}\ket{0}^{\otimes K}$, so linearity gives $U$.
\end{proof}

\subsection{Unitaries with unit row and column sums}

We next use the available diagonal gates to obtain another restriction of the family of unitaries that we need to implement. The Idel--Wolf theorem lets us multiply an arbitrary unitary on the left and right by diagonal unitaries so that every row and column sums to 1. Since we can implement the diagonal factors in constant depth, implementing unitaries with these row and column sums is sufficient.

\begin{prop}
\label{row and column sums}
If unbounded size $\qaczero$ circuits can implement every unitary whose rows and columns all sum to $1$, then they can implement an arbitrary unitary.
\end{prop}
\begin{proof}
Given a unitary $U$, we use the Idel--Wolf theorem \cite[Theorem 2]{idel2015sinkhorn} to choose diagonal unitaries $L,R$ such that $A=LUR$ satisfies
\[
\sum_j A_{ij}=1=\sum_i A_{ij}.
\]
We then implement $U=L^\dagger A R^\dagger$ by composing the assumed implementation of $A$ with the diagonal circuits from Proposition \ref{diagonal phase unitaries in constant depth}.
\end{proof}

\subsection{An alternative gateset for unbounded size \texorpdfstring{$\qaczero$}{QAC0}}

The operations we have already constructed give another description of the circuit model. We allow arbitrary diagonal unitaries and permutations of computational basis states as individual layers, together with layers of Hadamard gates. We show that this gives exactly the same families of implementable unitaries as unbounded size $\qaczero$.

\begin{dfn}
\label{HPP0 circuits}
An $\mathsf{HPP^0}$ circuit has a constant number of layers, each of which is a Hadamard gate on every qubit, an arbitrary diagonal unitary, or an arbitrary permutation of computational basis states. Circuit families have a common constant bound on the number of layers, with unrestricted width.
\end{dfn}

\begin{dfn}
\label{clean quantum implementation}
Let $U$ be a unitary on $m$ qubits and let $C$ be a circuit on $n\geq m$ qubits. We say that $C$ realises $U$ at distinct locations $(b_1,\ldots,b_m)\in[n]^m$ if
\[
C\ket{x}=(U,f)\ket{x}
\]
for every $x\in\zo^n$ that is zero outside these locations. Here $f(i)=b_i$, and $(U,f)$ acts as $U$ on the ordered qubits $b_1,\ldots,b_m$ and as the identity on all other qubits. By linearity, the equality holds for every state of the selected qubits when the other qubits are initially zero. Those other qubits also end in zero, as required by the implementation convention of Section~\ref{circuit-conventions}.
\end{dfn}

\begin{prop}
\label{HPP equivalence}
The same families of unitaries can be implemented by unbounded size $\qaczero$ circuits and $\mathsf{HPP^0}$ circuits.
\end{prop}
\begin{proof}
We first convert a $\qaczero$ circuit of depth $d$ into an $\mathsf{HPP^0}$ circuit. Each layer contains single-qubit gates and Toffoli gates acting on disjoint sets of qubits, so we can separate it into one layer of each type.

By Euler decomposition, we can write every single-qubit unitary as $P_1HP_2HP_3$, where $P_1,P_2,P_3$ are diagonal single-qubit unitaries that include any overall phase. This expresses a layer of single-qubit gates as a product of three diagonal layers and two Hadamard layers. On an unused qubit, we set $P_1=P_2=P_3=I$, so the two Hadamards cancel. Both Hadamard layers can then act on every qubit, as required by Definition \ref{HPP0 circuits}.

A tensor product of Toffoli gates and identities is a permutation of computational basis states. Each original layer requires at most six $\mathsf{HPP^0}$ layers, giving total depth at most $6d$.

Conversely, we convert an $\mathsf{HPP^0}$ circuit of constant depth $c$ into a $\qaczero$ circuit by replacing each layer. Hadamard layers are elementary in $\qaczero$. Proposition \ref{diagonal phase unitaries in constant depth} gives a $\qaczero$ implementation of a diagonal layer with a constant depth bound $d_1$. Proposition \ref{Any permutation in constant depth}, together with the fanout implementation discussed in Section \ref{quantum circuits}, gives a constant bound $d_2$ for a permutation layer. These implementations return their additional qubits to zero, so we can compose them without changing the action of later layers. We apply the Hadamard layers only to the original registers. The resulting circuit has depth at most $c\max(1,d_1,d_2)$ and the required action on those registers.
\end{proof}

\subsection{Circuits with adaptive intermediate measurements}

We have so far used only unitary circuits. For the remaining two propositions, we also allow computational basis measurements between gate layers and let later layers depend on the measurement outcomes. Depth counts both gate layers and measurement rounds, but excludes the classical computation used to choose subsequent layers. We require the overall quantum output to equal $U$ applied to the input for every possible sequence of intermediate measurement outcomes, including when the input is entangled with another system. The measurement records and measured qubits may be discarded. This differs from the usual unitary implementation convention, in which all additional qubits return to zero.

We first apply this measurement model to the Clifford hierarchy. Its first level consists of Pauli operators. For $l\geq2$, a unitary belongs to level $l$ if conjugating any Pauli operator by it gives a unitary in level $l-1$. Using gate teleportation \cite{gottesman1999teleportation}, we reduce the required correction by one level at each stage and obtain a depth bound proportional to $l$.

\begin{prop}
\label{Clifford hierarchy implementation}
Fix a positive integer $l$. Every $n$-qubit unitary in level $l$ of the Clifford hierarchy has an adaptive implementation of depth $O(l)$ using arbitrary single-qubit gates, $\toff$, $\fanout$ and computational basis measurements. A sufficient bound on both the number of additional qubits and the number of gates is $2^{O(nl)}$, which is exponential in $n$ for fixed $l$.
\end{prop}
\begin{proof}
We use induction on $l$. At level one, a Pauli operator is a tensor product of single-qubit gates up to an overall phase. We include that phase in one of the gates and implement the operator in one layer.

For the induction step, we prepare the Choi state
\begin{equation}
(I\otimes U)\frac{1}{\sqrt{2^n}}
 \sum_{x\in\zo^n}\ket{x}\ket{x}.
\end{equation}
This is a known pure state on $2n$ qubits. We can prepare it using Proposition \ref{Any state in constant depth} in a constant depth $c_1$ with $2^{O(n)}$ qubits and gates, without assuming an implementation of $U$ on an arbitrary input.

We perform Bell measurements between the input and the first half of this state, using CNOT gates, Hadamard gates and computational basis measurements in parallel. These measurements take a constant depth $c_2$. For each outcome, the unnormalised state of the remaining register is
\[
2^{-n}UP\ket{\psi},
\]
up to an overall phase, for a Pauli operator $P$ determined by the outcome. Each outcome has probability $4^{-n}$. The correction
\[
UP^\dagger U^\dagger
\]
lies in level $l-1$ and turns this state into $U\ket{\psi}$ after normalisation. By the induction hypothesis, we can implement it in depth at most $c(l-1)$, where we choose $c\geq\max(1,c_1+c_2)$. The total depth is at most $c_1+c_2+c(l-1)\leq cl$. The calculation remains valid when the input is entangled with a reference system on which the circuit acts as the identity. Thus the corrected output has the required action on entangled inputs as well.

Each state preparation uses $2^{O(n)}$ qubits and gates, and each measurement round has $4^n$ possible outcomes. We use at most $l-1$ rounds. Even if we allocate separate resources for every possible sequence of outcomes, the total is at most $2^{O(nl)}$ qubits and gates. Alternatively, we can prepare the state required at the next stage after the preceding outcomes are known. Both counts are for circuits with fanout included in the gate set.
\end{proof}

\begin{cor}
    Every generalised semi-Clifford gate can be implemented by unbounded size $\qaczero$ with intermediate measurements.
\end{cor}
\begin{proof}
    Since generalised semi-Cliffords are products of Clifford gates, permutation gates, and diagonal gates, this follows from Proposition \ref{Any permutation in constant depth}, Proposition \ref{diagonal phase unitaries in constant depth}, and Proposition \ref{Clifford hierarchy implementation}.
\end{proof}

\begin{prop}
\label{traceless involutions}
Let $S$ be the set of traceless unitary involutions. If every element of $S$ has an implementation by unbounded size $\qaczero$ circuits with adaptive measurements and a common constant depth bound, then every unitary has such an implementation.
\end{prop}
\begin{proof}
Let $U$ be an arbitrary unitary. We prepare its Choi state in unbounded size $\qaczero$ using Proposition \ref{Any state in constant depth}, replacing fanout gates by their exact implementations if necessary. We then perform gate teleportation as in the preceding proof.

We choose the Pauli representative for each outcome in $\{I,X,Y,Z\}^{\otimes n}$, absorbing its overall phase into that measurement branch. Then $P^\dagger=P$ and $P^2=I$. If $P=I$, we need no correction. Otherwise $\operatorname{tr}P=0$, so the correction $UPU^\dagger$ is a traceless unitary involution. By hypothesis, we can implement this correction with a common constant depth bound. Selecting the correction according to the measurement outcome gives the required adaptive circuit.
\end{proof}

We can also give a direct unitary reduction using one additional clean qubit. For any unitary $U$, the block matrix
\[
\begin{pmatrix}
0&U^\dagger\\
U&0
\end{pmatrix}
\]
is Hermitian, has trace zero and squares to the identity. Moreover,
\[
(X\otimes I)
\begin{pmatrix}
0&U^\dagger\\
U&0
\end{pmatrix}
\bigl(\ket{0}\ket{\psi}\bigr)
 =\ket{0}\,U\ket{\psi}.
\]
Given a unitary implementation of this involution, we follow it with $X$ on the extra qubit. The resulting circuit implements $U$ and returns that qubit to zero. Any additional qubits used to implement the involution also return to zero under Definition \ref{clean quantum implementation}. The reduction uses one extra qubit beyond the working qubits of the involution circuit.

\paragraph{Concluding remarks.}
We have proved that each of these tasks is sufficient for arbitrary unitary implementation. Each is also necessary if every unitary has an implementation with a common constant depth bound. Every displayed map sends an orthonormal set of inputs to an orthonormal set of outputs, so it extends to a unitary on all the registers. For copying, the first copy already makes the output states orthogonal. For the permuted lists, inputs $i$ and $k\ne i$ give orthogonal states in row $i$, since $\sigma_i(i)=i$ and $\sigma_k(i)\ne i$. For basis labels, the labels themselves are orthogonal. Thus each task is a special case of unitary implementation, and all the reductions use unitary circuits.

\subsection{Preparing pure states}

We prepare quantum states using the same probabilities and prefix
values. Single-qubit rotations replace the stochastic gates.
After computing the prefixes, we apply inverse rotations to the
qubits following the first 1. These operations return the later
qubits to zero while preserving the amplitudes associated with
each possible position of the first 1.

Both depth bounds are independent of the input length and the
target distribution or state. The elementary gate counts are
$O(n2^n)$, while the total number of gate inputs is $O(4^n)$.

For the quantum construction we may replace fanout by exact
circuits using only single-qubit and generalised Toffoli gates.
This preserves constant depth. However, a fanout in the prefix
circuit can act on $2^n-1$ qubits, so this replacement can give
a doubly exponential size bound. The $O(n2^n)$ gate bound applies
when fanout is included in the gate set.

\section{Towards constant-depth unitary implementation via port-based teleportation}
\label{sec-pbt-fixed}

Our state preparation results allow us to prepare the Choi state
of any unitary in constant depth. The remaining question is how
to use such a state to apply the unitary to an unknown input.
In the gate teleportation construction of
Proposition~\ref{Clifford hierarchy implementation}, a measurement
result produces $UP\ket{\psi}$, where $P$ is a Pauli operator
determined by that result. Obtaining $U\ket{\psi}$ then requires
the correction $UP^\dagger U^\dagger$. For a general $U$, our
constructions do not provide a constant-depth implementation
of this correction, so we turn to another method, which requires no corrections depending on $U$.

Port-based teleportation (PBT) removes the requirement to correct the teleported state at the expense of transmitting quantum information imperfectly~\cite{pbt-ishizaka2008}. In this protocol,  Alice performs a joint measurement on the state to be teleported and her halves of maximally entangled pairs. The measurement result identifies one of Bob's registers, called a port, which he selects as the output. Replacing the entangled pairs by copies of the Choi
state of $U$ gives an approximate implementation of $U$ using
the same measurement on Alice's registers. 

The resource state (typically a collection of maximally entangled pairs) is a known pure state and can be prepared exactly
in constant depth. We place the selected port in a fixed output register by swapping it with that register while retaining the port label. This operation is a permutation of computational basis states, so our permutation construction implements it in constant depth. We thus only need to determine the depth required to implement the initial measurement. We will need to realize unitarily
and retain all the qubit registers so that the circuit and its inverse remain available for amplitude amplification.

We construct a protocol whose depth is $O(\sqrt d)$ for an input
of dimension $d$, independently of the number of ports. For every
fixed $d$ we can make its entanglement fidelity arbitrarily close
to one by increasing the number of ports without increasing the
depth bound. This establishes the result for each fixed input
dimension and we believe that this dependence on $d$ cannot be further improved with existing techniques. We leave open the question of optimality of this the dependence on $d$. Throughout this section we allow for an unbounded number of ancillae.

\subsection{Teleportation circuit and fidelity bound}

Let $C$ be the input state of dimension $d$. Alice and Bob share
$M$ maximally entangled pairs
\begin{equation}
    \ket{\Omega}
    =\bigotimes_{i=1}^M\ket{\Phi_d}_{A_iB_i},
    \qquad
    \ket{\Phi_d}=\frac1{\sqrt d}\sum_{a=0}^{d-1}\ket{a,a}.
    \label{pbt-resource}
\end{equation}
Alice holds $C,A_1,\ldots,A_M$ and Bob holds $B_1,\ldots,B_M$.
Here $d$ refers to the input register $C$ whose state we
wish to teleport. Each register $A_i$ and $B_i$ also has dimension
$d$. For an $n$-qubit input, $d=2^n$. If $d$ is not a power of two,
we represent each register by $\lceil\log_2 d\rceil$ qubits and use
only the computational basis states $\ket{0},\ldots,\ket{d-1}$.
All the identities below apply to the subspaces spanned by these states. Our circuit is defined on the full qubit registers, but its action outside these subspaces does not affect the teleportation protocol.

Alice holds the input register $C$ and the registers
$A_1,\ldots,A_M$. Bob holds the output ports $B_1,\ldots,B_M$.
Each pair $A_iB_i$ is initially in the maximally entangled state
$\ket{\Phi_d}$. We implement Alice's teleportation measurement by a unitary circuit acting on $C,A_1,\ldots,A_M$ and ancillary qubits.
The circuit records a port label. Together with the other
registers on Alice's side, this label specifies the output port.
For results that do not identify a port, we select port~1
so that the protocol always produces an output.
A permutation of computational basis states places the selected
port in a fixed output register. Tracing out all other registers
then defines the teleportation channel.

We identify the selected port with a fixed output register $B$ and
denote the resulting channel by $\mathcal E$. We assess its fidelity
by applying it to one half of a maximally entangled state while
leaving the reference $R$ unchanged. The entanglement fidelity is
\begin{equation}
    F_{\mathrm e}(\mathcal E)
    =\bra{\Phi_d}_{RB}
       (\operatorname{id}_R\otimes\mathcal E)
       (\ket{\Phi_d}\!\bra{\Phi_d}_{RC})
      \ket{\Phi_d}_{RB}.
    \label{pbt-fidelity-definition}
\end{equation}
We use the gates and depth convention of
Section~\ref{circuit-conventions}, with disjoint supports in every layer.
All additional qubits start in zero. The qubits discarded when defining
$\mathcal E$ need not return to zero. Our circuit uses no intermediate measurements. We prove the following result:

\begin{thm}
\label{pbt-main}
For every input dimension $d\geq2$ and $M\geq d^2-1$ ports,
there is a unitary circuit of depth $O(\sqrt d)$ that implements
deterministic PBT using only single-qubit
and generalised Toffoli gates. The PBT uses maximally entangled states~\eqref{pbt-resource}, and has entanglement fidelity
\begin{equation}
    F_{\mathrm e}(\mathcal E)
    \geq\left(1-\frac{d^2-1}{2M}\right)^2.
    \label{pbt-main-fidelity}
\end{equation}
The depth bound is independent of $M$ and includes resource
preparation and output selection. No intermediate measurements
are required.
\end{thm}

\begin{cor}
\label{pbt-accuracy}
For every fixed $d\geq2$ and $0<\varepsilon<1$, we obtain entanglement
fidelity at least $1-\varepsilon$ with
\begin{equation}
    M=\left\lceil\frac{d^2-1}{\varepsilon}\right\rceil
\end{equation}
and a depth bound independent of $\varepsilon$. For fidelity at least
$1/4$, it suffices to take $M=d^2-1$.
\end{cor}

\begin{proof}
In the first case, the right-hand side of
\eqref{pbt-main-fidelity} is at least
$(1-\varepsilon/2)^2\geq1-\varepsilon$. In the second case it equals
$1/4$.
\end{proof}

We prove the theorem in three steps.
\begin{enumerate}
    \item We construct Alice's measurement using two state preparation
circuits, with the second applied in reverse. Each preparation
depends on a computational basis label that is retained without
being measured. We choose the prepared states so that their
overlaps give the required measurement amplitudes. The total probability of success will approach $1/d$ as $M$ grows.
\item We apply amplitude amplification to Alice's registers,
using a number of steps that depends only on $d$. This increases
the probability of success but can change the conditional output
state, so we must bound the teleportation fidelity.
    \item We bound its entanglement fidelity using the first two
    moments of a positive operator and a scalar inequality.
    The bound approaches one as $M$ grows.
\end{enumerate}
Figure~\ref{figure-pbt-process} illustrates the construction.
Both state preparation circuits have depth bounded independently
of $d$ and $M$ when circuit size and the number of ancillary qubits
are unrestricted. The number of amplification steps depends only
on $d$, giving a total depth bound independent of $M$.
After proving the fidelity bound, we detail the supporting circuit
constructions and the associated depth calculation in
Section~\ref{pbt-depth-section}.

\begin{figure}[tbp]
    \centering
    \resizebox{\linewidth}{!}{%
\begin{tikzpicture}[
    x=1cm,y=1cm,
    line width=0.65pt,
    every node/.style={font=\small},
    wire/.style={-{Stealth[length=2mm]},draw=black!80},
    resource/.style={draw=black!55,fill=black!3,rounded corners=2pt},
    alice/.style={draw=blue!55!black,fill=blue!4,rounded corners=2pt},
    select/.style={draw=teal!65!black,fill=teal!5,rounded corners=2pt},
    stage/.style={draw=blue!55!black,fill=blue!4,rounded corners=2pt,
                  minimum height=1.05cm,align=center,inner sep=5pt}
]
    \node[anchor=west,font=\small\bfseries] at (0,5.65)
        {(a) Teleportation circuit};

    \draw[resource] (0, -0.55) rectangle (2.45,2.55);
    \node at (1.225,2.05) {$A_1,\ldots,A_M$};
    \node[align=center] at (1.225,1.1)
        {$M$ shared pairs\\[3pt]$\ket{\Omega}$};
    \node at (1.225,0) {$B_1,\ldots,B_M$};

    \node[align=center] (input) at (1.225,3.55)
        {Unknown input\\[2pt]$C$};
    \node[align=center] (work) at (5.8,5.0)
        {Ancillary qubits\\[2pt]$\ket{0\cdots0}$};

    \draw[alice] (4.2,0.8) rectangle (7.4,4.0);
    \node[font=\small\bfseries] at (5.8,3.55) {Alice};
    \node[font=\large] at (5.8,2.92) {$U_q$};
    \node[align=center,text width=2.7cm] at (5.8,1.9)
        {Apply $V$, then\\[3pt]$q$ amplification\\steps};

    \draw[wire] (input.east) -- (4.2,3.55);
    \draw[wire] (work.south) -- (5.8,4.0);
    \draw[wire] (2.45,2.05) -- (4.2,2.05);

    \draw[select] (10.1,-0.55) rectangle (12.7,3.35);
    \node[align=center,font=\small\bfseries] at (11.4,2.63)
        {Select output\\port};
    \node[align=center,text width=2.2cm] at (11.4,1.2)
        {Port $i$\\on success\\[6pt]Port $1$\\otherwise};
    \draw[wire] (7.4,2.05) -- (10.1,2.05)
        node[midway,above=5pt,align=center] {Alice registers};
    \draw[wire] (2.45,0) -- (10.1,0)
        node[midway,below=5pt] {Bob's ports};

    \draw[wire] (12.7,2.05) -- (13.35,2.05);
    \node[anchor=west,align=left] at (13.4,2.05)
        {Discard\\remaining\\registers};
    \draw[wire] (12.7,0) -- (13.35,0);
    \node[anchor=west,align=left] at (13.4,0) {Output $B$};
    \node[anchor=west,font=\footnotesize] at (13.4,-0.48)
        {No correction};

    \node[anchor=west,font=\small\bfseries] at (0,-1.65)
        {(b) The initial circuit $V$};
    \node[stage,minimum width=3.8cm] (first) at (2.4,-2.75)
        {First preparation};
    \node[stage,minimum width=4.4cm] (second) at (7.6,-2.75)
        {Reverse second\\preparation};
    \node[stage,minimum width=3.4cm] (flag) at (12.7,-2.75)
        {Single-qubit\\rotation};
    \draw[wire] (first.east) -- (second.west);
    \draw[wire] (second.east) -- (flag.west);
    \node[align=center,font=\footnotesize] at (7.6,-3.65)
        {$V$ and $V^\dagger$ are used throughout the amplification};
\end{tikzpicture}
}
    \caption{Unitary circuit for port-based teleportation.
(a) Alice and Bob share $M$ maximally entangled pairs.
Alice applies $U_q$ to the input register $C$, her registers
$A_1,\ldots,A_M$ and ancillary qubits. A permutation of
computational basis states places the selected port in a fixed
output register. The circuit selects port $i$ when the success
condition holds and port~1 otherwise. Tracing out all other
registers gives the channel $\mathcal E$. The arrows represent
quantum registers. The circuit uses no intermediate measurements
and applies no correction to the selected port.
(b) We construct $V$ by applying the first conditional state
preparation, the inverse of the second, and a single-qubit
rotation that adjusts the success amplitude. Amplitude
amplification uses $V$, $V^\dagger$ and the two reflections
defined in Section~\ref{pbt-amplification-section}.
Before amplification, every successful result gives the exact
input state at the selected port. After amplification, the
output can differ from the input, with entanglement fidelity
bounded as in Theorem~\ref{pbt-main}.}
    \label{figure-pbt-process}
\end{figure}
\clearpage

\subsection{The initial measurement}
\label{pbt-initial-section}

We first construct a measurement for which certain results
teleport the input exactly to an identified port. We call these
results successful. This measurement is an intermediate step
in the construction. The final deterministic protocol assigns
an output port to every result, including those outside the
successful set.

For a given port $i$, projecting $C,A_i$ onto
$\ket{\Phi_d}$ transfers the input state to $B_i$ without a
correction. We describe this operation by
\begin{equation}
    T_i=\bra{\Phi_d}_{CA_i}\otimes I_{A_{\ne i}},
    \label{pbt-measurement-operator}
\end{equation}
where $A_{\ne i}$ denotes all Alice registers $A_j$ with
$j\ne i$. After the projection, the registers $C,A_i$ are in
the fixed state $\ket{\Phi_d}$. The operator $T_i$ describes
the unnormalised state of the remaining registers, with this
fixed factor omitted.

The teleportation identity is
\begin{equation}
    (\bra{\Phi_d}_{CA_i}\otimes I_{B_i})
       (\ket{\psi}_C\ket{\Phi_d}_{A_iB_i})
    =\frac1d\ket{\psi}_{B_i}.
    \label{pbt-teleportation-identity}
\end{equation}
After normalisation, the state at $B_i$ is exactly the input
state. This identity also holds when the input is entangled
with a reference, with the identity acting on that reference.

We introduce a register for the port label and combine these
operators into $T=\sum_{i=1}^M\ket{i}T_i$.

The output contains the label $i$ and the remaining $M-1$
Alice registers in their original order. Let  $P_i=\ket{\Phi_d}\!\bra{\Phi_d}_{CA_i}\otimes I_{A_{\ne i}},
     S=\sum_{i=1}^M P_i,
     L=\frac{M+d-1}{d}$.
This makes $T_i^\dagger T_i=P_i$ and $T^\dagger T=S$.

We will construct a measurement whose successful operators
are $T_i/\sqrt L$. Equivalently, the part of the unitary circuit
associated with successful results implements $T/\sqrt L$.
The construction will also establish $S/L\leq I$, which is necessary
for these operators to form a measurement.

To construct the circuit, we first express $T$ in the
computational basis. For an input label
$x=(c;a_1,\ldots,a_M)$,
\begin{equation}
    T\ket{c;a_1,\ldots,a_M}
    =\frac1{\sqrt d}\sum_{i:a_i=c}\ket{i;a_{\ne i}}.
    \label{pbt-basis-action}
\end{equation}
Each term corresponds to a position where $a_i=c$. The
coefficient $1/\sqrt d$ is the amplitude of $\ket{c,c}$ in
$\ket{\Phi_d}$.

We obtain the required amplitudes from two state preparation
circuits, with the second applied in reverse. Each preparation
depends on a computational basis label that is retained
without being measured. The following lemma gives the
required depth bound.

\begin{lem}
\label{pbt-conditional-preparation}
For every finite family of normalised states
$\{\ket{\psi_x}:x\in\zo^r\}$ on $m\geq1$ qubits with nonnegative
computational basis coefficients, the map
\begin{equation}
    \ket{x}\ket{0^m}\longmapsto\ket{x}\ket{\psi_x}
    \label{pbt-conditional-map}
\end{equation}
has an exact implementation of depth at most $42$ with fanout
included in the gate set. Every additional work qubit is returned to $|0\rangle$.
\end{lem}

\begin{proof}
Write
\begin{equation*}
    \ket{\psi_x}
    =\sum_{j=0}^{2^m-1}\sqrt{p_j(x)}\ket{j}.
\end{equation*}
We retain $x$ and apply the preparation construction of
Theorem~\ref{Any state in constant depth} with these
probabilities. Its rotation parameters are defined even
when a denominator is zero, so no relabelling is needed
when $p_0(x)=0$.

Only the two rotation stages depend on $x$. In the first
stage we apply a product of real single-qubit rotations
whose angles depend on $x$. These rotations share a fixed
eigenbasis. We change each target qubit to this basis,
apply a diagonal unitary on the label and target registers,
and reverse the basis changes.
Proposition~\ref{diagonal phase unitaries in constant depth}
implements the diagonal unitary in seven layers, giving
nine layers for this stage. Treating the whole stage as
one diagonal unitary allows all rotation angles to depend
on the same label while respecting the requirement that
gates in a layer have disjoint supports.

The second rotation stage applies inverse rotations
conditioned on the prefix bits and $x$. The same basis
changes make this operation diagonal on the label, prefix
and target registers. This stage also takes nine layers.
The phases of both diagonal unitaries are determined by
the known family of states. Their calculation is not part
of the quantum circuit.

The prefix computations, additional indicator and decoding
do not depend on $x$. For each computational basis label
the preparation proof gives the required state and returns
all work qubits to zero. Linearity gives the required
action on superpositions of labels. The count in
Section~\ref{pbt-preparation-depths} gives total depth
at most $42$.
\end{proof}

Retaining the label makes the outputs for distinct labels
orthogonal even when the prepared states $\ket{\psi_x}$
are not orthogonal.

We now apply the lemma twice. The first preparation retains
the input label $x$ and prepares a superposition of the
possible output labels $y=(i;a_{\ne i})$ in
\eqref{pbt-basis-action}. The second retains $y$ and prepares
a superposition of compatible input labels. Applying the
inverse of the second preparation gives amplitudes equal
to the overlaps between the two prepared states. We choose
their coefficients so that every compatible pair of labels
has the same overlap.

\begin{lem}
\label{pbt-factorisation}
Two conditional state preparation circuits, with the second
applied in reverse, give the measurement operator $T/\sqrt L$
for successful results. Each preparation has depth at most
$42$ with fanout included in the gate set.
\end{lem}

\begin{proof}
We use registers for the input label
$x=(c;a_1,\ldots,a_M)$ and output label $y=(i;a_{\ne i})$,
together with an ancillary qubit $f$. Let $r(x)$ be the
number of positions satisfying $a_i=c$.

For $r(x)>0$, the first preparation is
\begin{equation}
    V_{\mathrm{in}}\ket{x}
    =\ket{x}\frac1{\sqrt{r(x)}}
       \sum_{i:a_i=c}\ket{i;a_{\ne i}}\ket{0}_f.
    \label{pbt-Vin}
\end{equation}
It retains $x$ and assigns equal amplitude to each matching
position. When $r(x)=0$, we prepare any fixed output label
and set $f=1$. We use the same choice for unused input
strings in the qubit encoding. Every prepared state is
normalised, and retaining $x$ makes $V_{\mathrm{in}}$
an isometry.

For the second preparation, fix an output label
$y=(i;a_{\ne i})$. A compatible input must have $c=a_i=b$
for some $b\in\{0,\ldots,d-1\}$. Let $x_b(y)$ denote the
input label obtained by inserting $b$ at position $i$
and setting $c=b$. We define
\begin{equation}
    V_{\mathrm{out}}\ket{y}
    =\sum_{b=0}^{d-1}
       \sqrt{\frac{r(x_b(y))}{M+d-1}}\,
       \ket{x_b(y)}\ket{y}\ket{0}_f.
    \label{pbt-Vout}
\end{equation}
To check normalisation, sum $r(x_b(y))$ over $b$.
The inserted position contributes one match for each
of the $d$ values of $b$. Each of the remaining $M-1$
positions contributes once. Hence
\begin{equation*}
    \sum_{b=0}^{d-1}r(x_b(y))=M+d-1.
\end{equation*}
For unused output labels, we prepare any fixed input
basis state and keep $f=0$. Retaining $y$ makes
$V_{\mathrm{out}}$ an isometry. Both maps prepare known
states with nonnegative coefficients while retaining a
basis label, so Lemma~\ref{pbt-conditional-preparation}
applies to each.

Suppose $y$ is obtained from $x$ by choosing a position
where $a_i=c$ and omitting $a_i$. Then $x_c(y)=x$, and
the overlap between the two prepared states is
\begin{equation}
    \frac1{\sqrt{r(x)}}
       \sqrt{\frac{r(x)}{M+d-1}}
    =\frac1{\sqrt{M+d-1}}.
    \label{pbt-cancellation}
\end{equation}
The dependence on the number of matching positions cancels.
All other overlaps vanish. Since
$dL=M+d-1$, comparison with~\eqref{pbt-basis-action} gives
\begin{equation}
    V_{\mathrm{out}}^\dagger V_{\mathrm{in}}
       =\frac{T}{\sqrt L}.
    \label{pbt-factorisation-equation}
\end{equation}
Both maps are isometries, so their product has operator
norm at most one. It follows that $S/L\leq I$.

We implement these maps by unitary preparation circuits.
Initially the output register, $f$ and all work qubits
are zero. We apply the circuit for $V_{\mathrm{in}}$
and then the inverse of the circuit for
$V_{\mathrm{out}}$.

A result is successful when the input register, $f$ and
all preparation work qubits are zero after these two
circuits. The output label register is unrestricted.
The operator associated with these zero values is
$V_{\mathrm{out}}^\dagger V_{\mathrm{in}}$, as required.
We must include all work qubits in this condition because
the inverse preparation can receive a state outside the
range of $V_{\mathrm{out}}$.

The circuit uses no intermediate measurement. The zero
conditions specify the subspace about which we reflect
during amplitude amplification. We implement this
reflection directly, without recording success in another
ancillary qubit.

For an unused input string, the first preparation has
$f=1$, whereas every state prepared by
$V_{\mathrm{out}}$ has $f=0$. Their overlap is zero.
Thus the successful measurement operator vanishes outside
the logical input subspace, and its adjoint maps into
that subspace.
\end{proof}

For example, take $d=2$, $M=3$ and
$x=(0;0,0,0)$. All three positions match, so the first
preparation assigns amplitude $1/\sqrt3$ to each possible
output label. For $y=(1;0,0)$, the compatible inputs are
$x_0(y)=(0;0,0,0)$ and $x_1(y)=(1;1,0,0)$.
They have three matches and one match respectively.
The second preparation assigns them amplitudes
$\sqrt3/2$ and $1/2$. The overlap corresponding to the
original input is $\frac1{\sqrt3}\frac{\sqrt3}{2}
    =\frac12
    =\frac1{\sqrt{M+d-1}}$.
The same cancellation gives this coefficient for every
compatible input and output label.

We now calculate the probability of a successful result
when Alice and Bob share the resource~\eqref{pbt-resource}.

\begin{lem}
\label{pbt-initial-success}
The initial measurement has total success probability
\begin{equation}
    \mu=\frac{M}{d^2L}
       =\frac{M}{d(M+d-1)},
    \label{pbt-mu}
\end{equation}
independently of the input state. For every successful
result $i$, teleportation to $B_i$ is exact.
\end{lem}

\begin{proof}
The reduced state on Alice is
$\rho_C\otimes I_{A_1\cdots A_M}/d^M$.
The expectation of each $P_i$ in this state is $1/d^2$.
The probability of success is the expectation of $S/L$,
which gives $M/(d^2L)$.

For successful result $i$, the measurement operator is
$T_i/\sqrt L$. Equation~\eqref{pbt-teleportation-identity}
shows that the normalised state at $B_i$ equals the input
state. This also holds for inputs entangled with a
reference.
\end{proof}

For $d\geq2$, we have $\mu<1/d$, with $\mu$ tending to
$1/d$ as $M$ increases. Increasing the number of ports
alone does not make this probability approach one.
We next apply amplitude amplification using the circuit
and its inverse. This changes the successful measurement
operators, so we then bound the entanglement fidelity
of the resulting protocol.

\subsection{Amplitude amplification and output selection}
\label{pbt-amplification-section}

We apply amplitude amplification~\cite{pbt-brassard2002} to
the circuit constructed in Section~\ref{pbt-initial-section},
using operations on Alice's registers. Amplitude
amplification has also been used in other circuits for
port-based teleportation~\cite{pbt-wills2024}.

The effect of amplification depends on the eigenvalues of
$S$, which acts on $C,A_1,\ldots,A_M$. For a unit eigenvector
with eigenvalue $s$, the initial amplitude in the successful
subspace is $\sqrt{s/L}$. These amplitudes can differ even
though the total success probability $\mu$ on the shared
resource is independent of the state to be teleported.

If an initial success amplitude is $\sin\theta$, then $q$
amplification steps change it to $\sin((2q+1)\theta)$.
We choose the parameters so that this expression equals one
at $s=M/d^2$, the mean of the eigenvalues of $S$ on Alice's
logical input space. Indeed, each $P_i$ has rank $d^{M-1}$
on a space of dimension $d^{M+1}$, so the mean of the
eigenvalues of $S=\sum_iP_i$ is $M/d^2$.

We first reduce the initial amplitudes by a common factor.
This allows us to choose an integer number of amplification
steps that depends only on $d$.

Let $k$ be the smallest odd integer satisfying
\begin{equation}
    k\geq\frac\pi2\sqrt{d+1},
    \label{pbt-iterations}
\end{equation}
and set $q=(k-1)/2$. Define
\begin{equation}
    a=\sin\frac{\pi}{2k},
    \qquad c=\frac{a}{\sqrt\mu}.
    \label{pbt-attenuation}
\end{equation}
At $s=M/d^2$, the initial amplitude is $\sqrt\mu$.
Multiplication by $c$ reduces it to $a$, for which
\begin{equation*}
    \sin\bigl(k\arcsin(a)\bigr)=1.
\end{equation*}

We check that $c\leq1$, so this reduction can be implemented
by a single-qubit rotation. Since
$M\geq d^2-1\geq d(d-1)$, equation~\eqref{pbt-mu} gives
$\mu\geq1/(d+1)$. Our choice of $k$ then gives
\begin{equation*}
    a\leq\frac{\pi}{2k}
      \leq\frac1{\sqrt{d+1}}
      \leq\sqrt\mu.
\end{equation*}
Thus $0<c\leq1$.

We prepare an additional ancillary qubit in
$c\ket{1}+\sqrt{1-c^2}\ket{0}$ and include the value 1
of this qubit in the success condition. Let $V$ denote
the two preparation circuits from
Section~\ref{pbt-initial-section}, with the second applied
in reverse, followed by this rotation. Its operator for
successful results is $cT/\sqrt L$. The rotation can depend
on $M$, but the number $q$ of amplification steps depends
only on $d$.

We now define the two reflections used in amplification.
Let $\Pi_{\mathrm{in}}$ project onto zero in every additional
Alice register. It acts as the identity on
$C,A_1,\ldots,A_M$, whose states need not be known.
Let $\Pi_{\mathrm s}$ project onto the successful subspace:
the registers $C,A_1,\ldots,A_M$, the qubit $f$ and all
preparation work qubits must be zero, and the additional
qubit introduced above must be 1. The output label register
is unrestricted.

We apply $V$ followed by $q$ amplification steps:
\begin{equation}
    R_{\mathrm{in}}=I-2\Pi_{\mathrm{in}},\qquad
    R_{\mathrm s}=I-2\Pi_{\mathrm s},\qquad
    U_q=(-V R_{\mathrm{in}}V^\dagger R_{\mathrm s})^qV.
    \label{pbt-amplification}
\end{equation}
Each reflection applies phase $-1$ precisely when its
register conditions hold. Both act only on Alice's
registers and have depth at most three by Section~\ref{pbt-conditional-operations}.

The following lemma gives the measurement operators after
amplification.

\begin{lem}
\label{pbt-amplification-operators}
After applying $U_q$, the measurement operator for
successful result $i$ is
\begin{equation}
    K_i=T_iG(S),
    \qquad
    G(s)=
    \frac{\sin\!\left(k\arcsin\!\left(a\sqrt{d^2s/M}\right)\right)}
         {\sqrt s}.
    \label{pbt-K}
\end{equation}
The displayed expression applies to positive eigenvalues
$s$ of $S$, with the continuous value
$G(0)=ka d/\sqrt M$. The function $G$ extends to a real
polynomial in $s$.
\end{lem}

\begin{proof}
Let $\ket{v}$ be a unit eigenvector of $S$ with eigenvalue
$s>0$. After applying $V$, its component in the successful
subspace is
\begin{equation*}
    \frac{c}{\sqrt L}T\ket{v}.
\end{equation*}
Since $T^\dagger T=S$, the normalised state of this component
is $T\ket{v}/\sqrt s$, and its amplitude is
\begin{equation*}
    c\sqrt{s/L}=a\sqrt{d^2s/M}.
\end{equation*}
This amplitude is at most one because $S/L\leq I$ and
$c\leq1$.

Write the amplitude as $\sin\theta$. The two reflections
preserve the span of the successful and unsuccessful
components of $V\ket{v}$. Each amplification step increases
the angle by $2\theta$. After $q$ steps, the successful
component is
\begin{equation*}
    \sin(k\theta)\frac{T\ket{v}}{\sqrt s}.
\end{equation*}
This is $TG(S)\ket{v}$. Selecting the port label $i$
gives $K_i=T_iG(S)$.

For $s=0$, we have $T\ket{v}=0$, and amplification leaves
the successful component zero. The limit of the displayed
expression for $G(s)$ as $s$ tends to zero is
$ka d/\sqrt M$.

Finally, $k$ is odd, so $\sin(k\arcsin t)$ is an odd
polynomial in $t$. Substituting
$t=a\sqrt{d^2s/M}$ and dividing by $\sqrt s$ gives a
polynomial in $s$.
\end{proof}

The polynomial $G$ describes the action of the circuit.
Its implementation requires only $V$, $V^\dagger$ and
the two reflections. We do not diagonalise $S$ or calculate
its eigenvalues within the circuit.

To complete the deterministic protocol, we assign an output
port to every measurement result. When the success
condition holds and the label is
$i\in\{1,\ldots,M\}$, we select port $i$. For every other
result, we select port~1. All results are included in
the teleportation channel, and no correction is applied
to the selected port.

In the unitary circuit, we perform this selection by a
permutation of computational basis states. It swaps the
selected port with a fixed output register while preserving
Alice's registers. The successful and remaining components
thus stay orthogonal. Tracing out all registers except
the selected output gives the deterministic teleportation
channel.

Before amplification, each successful measurement operator
was proportional to $T_i$ and gave exact teleportation.
After amplification, the factor $G(S)$ can change the
conditional output state. We next bound the entanglement
fidelity of the resulting protocol.

\subsection{Entanglement fidelity}
\label{pbt-fidelity-section}

We now bound the entanglement fidelity of the deterministic
protocol. Let $F_{\mathrm s}$ denote the contribution from
successful results, including their probabilities. The
remaining results, which select port~1, contribute
nonnegatively, so $F_{\mathrm e}(\mathcal E)\geq F_{\mathrm s}$.
We bound $F_{\mathrm s}$ using the first two moments of
the eigenvalues of $S$.

In the definition of entanglement fidelity, the input $C$
is maximally entangled with a reference $R$. Together with
the shared resource, this makes the reduced state on
$C,A_1,\ldots,A_M$ maximally mixed. All traces below are
taken on this logical input space, whose dimension is
$D=d^{M+1}$. We write $\tau(H)=\operatorname{tr}(H)/D$
and $Y=d^2S/M$. Thus $\tau(H)$ is the expectation of $H$
in the state $I/D$, and the rescaling from $S$ to $Y$
makes the mean of the eigenvalues equal to one.

\begin{lem}
\label{pbt-moments}
We have
\begin{equation}
    \tau(Y)=1,\qquad
    \tau(Y^2)=1+\frac{d^2-1}{M},\qquad
    \tau((Y-I)^2)=\frac{d^2-1}{M}.
    \label{pbt-moments-equation}
\end{equation}
\end{lem}

\begin{proof}
Each $P_i$ is a projector of rank $D/d^2$, so
$\tau(P_i)=d^{-2}$.

For distinct ports $i,j$, applying the teleportation
identity twice gives $P_iP_jP_i=d^{-2}P_i$.
To see the factor explicitly, consider
$\ket{\Phi_d}_{CA_i}\ket{b}_{A_j}$. Applying $P_j$ gives
$d^{-1}\ket{b}_{A_i}\ket{\Phi_d}_{CA_j}$.
Applying $P_i$ then returns $d^{-2}$ times the original
state. These vectors span the relevant subspace, and
both projectors act as the identity on the other registers.

Taking the trace gives $\tau(P_iP_j)=d^{-4}$.
Since $P_i^2=P_i$, we obtain $\tau(S)=M/d^2$ and
$\tau(S^2)=M/d^2+M(M-1)/d^4$.
Rescaling by $d^2/M$ proves the first two identities.
The third follows from
$\tau((Y-I)^2)=\tau(Y^2)-1$.
\end{proof}

For fixed $d$, the variance of the eigenvalues of $Y$
decreases as $M$ grows. We now relate this variance to
the fidelity. For a measurement result with operator $K_i$, write
$E_i=K_i^\dagger K_i$. On Alice's state $\rho$, this
result occurs with probability $\operatorname{tr}(E_i\rho)$.
After tracing out her registers, its contribution
to fidelity depends only on $E_i$ and the selected port.

We express the bound using
$h_k(y)=\sqrt y\sin(k\arcsin(a\sqrt y))$
for $0\leq y\leq a^{-2}$, where $k$ and $a$ are given
by~\eqref{pbt-iterations} and~\eqref{pbt-attenuation}.
Our choice of amplification gives $h_k(1)=1$.
Since $a^2Y=c^2S/L\leq I$, this function is defined
on every eigenvalue of $Y$.

\begin{lem}
\label{pbt-fidelity-trace}
A measurement result with positive operator $E_i$ that
selects port $i$ contributes $\tau(P_iE_i)$ to entanglement
fidelity. For the successful results in~\eqref{pbt-K},
their total contribution satisfies
\begin{align}
    F_{\mathrm s}
      &=\sum_{i=1}^M
          \tau\bigl((P_iG(S)P_i)^2\bigr),
        \label{pbt-fidelity-exact}\\
    F_{\mathrm s}
      &\geq[\tau(h_k(Y))]^2.
        \label{pbt-fidelity-bound}
\end{align}
\end{lem}

\begin{proof}
For a result that selects port $i$, the fidelity definition
tests whether $R,B_i$ are in the state $\ket{\Phi_d}$.
This test acts on different registers from Alice's
measurement, so we can evaluate it first. We use
\begin{equation*}
    \bra{\Phi_d}_{RB_i}
       \bigl(\ket{\Phi_d}_{RC}\ket{\Omega}\bigr)
    =\frac1d\ket{\Phi_d}_{CA_i}
       \bigotimes_{j\ne i}\ket{\Phi_d}_{A_jB_j}.
\end{equation*}
Tracing out each remaining Bob register gives $I/d$
on its partner at Alice. The resulting unnormalised
operator on her input registers is $P_i/D$.
The contribution of the measurement result is thus
$\operatorname{tr}(E_iP_i/D)=\tau(P_iE_i)$.

For a successful result, $K_i=T_iG(S)$ and
$T_i^\dagger T_i=P_i$, so $E_i=G(S)P_iG(S)$.
Cyclicity of the trace gives
$\tau(P_iE_i)=\tau((P_iG(S)P_i)^2)$.
Summing over $i$ proves~\eqref{pbt-fidelity-exact}.

Each $P_iG(S)P_i$ is Hermitian and acts on the range
of $P_i$, which has dimension $D/d^2$.
Apply Cauchy--Schwarz to the eigenvalues of all these
operators together. There are $MD/d^2$ eigenvalues,
their sum is $\operatorname{tr}(SG(S))$, and the sum
of their squares is $DF_{\mathrm s}$. Hence
\begin{equation*}
    F_{\mathrm s}
    \geq\frac{d^2}{MD^2}
       \left[\operatorname{tr}(SG(S))\right]^2
    =\frac{d^2}{M}[\tau(SG(S))]^2.
\end{equation*}
The identity $SG(S)=(\sqrt M/d)h_k(Y)$ now
gives~\eqref{pbt-fidelity-bound}.
\end{proof}

It remains to bound $\tau(h_k(Y))$. A quadratic lower
bound on $h_k$ will suffice, since its trace can be
evaluated using Lemma~\ref{pbt-moments}.
For $k=3$, the triple-angle identity gives
$h_3(y)=(3y-y^2)/2$ on $0\leq y\leq4$.
The same quadratic is a lower bound for every odd
$k\geq3$ on the required domain.

\begin{lem}
\label{pbt-scalar-lemma}
Let $k\geq3$ be odd and put $a=\sin(\pi/(2k))$. Then
\begin{equation}
    h_k(y)\geq\frac{3y-y^2}{2}
    \qquad\text{for }0\leq y\leq a^{-2}.
    \label{pbt-scalar-inequality}
\end{equation}
\end{lem}

\begin{proof}
Write $z=\sqrt y$. For $0\leq z\leq2$, we compare
$k\arcsin(az)$ with $3\arcsin(z/2)$.
Consider $f_z(t)=\arcsin(z\sin t)$ for
$0\leq t\leq\pi/6$. In the interior of this interval,
\begin{equation*}
    f_z''(t)
    =\frac{z(z^2-1)\sin t}
           {(1-z^2\sin^2t)^{3/2}}.
\end{equation*}
Thus $f_z$ is concave for $0\leq z\leq1$ and convex
for $1\leq z\leq2$. Since $f_z(0)=0$, the ratio
$f_z(t)/t$ decreases with $t$ in the first case and
increases in the second. Endpoint cases follow
by continuity.

We compare $t=\pi/(2k)$ with $t=\pi/6$.
For $0\leq z\leq1$, concavity gives
$0\leq3\arcsin(z/2)\leq k\arcsin(az)\leq\pi/2$.
The final inequality follows from
$\arcsin(az)\leq\arcsin(a)=\pi/(2k)$.
Sine is increasing on this interval, so
$\sin(k\arcsin(az))\geq\sin(3\arcsin(z/2))$.

For $1\leq z\leq2$, convexity gives
$\pi/2\leq k\arcsin(az)\leq3\arcsin(z/2)\leq3\pi/2$.
Sine is decreasing on this interval, giving the same
comparison. In both cases,
$h_k(z^2)\geq h_3(z^2)=(3z^2-z^4)/2$.

For $2\leq z\leq a^{-1}$, the bound $\sin\theta\geq-1$
gives $h_k(z^2)\geq-z$. Moreover,
$-z-(3z^2-z^4)/2=z(z-2)(z+1)^2/2\geq0$.
This proves~\eqref{pbt-scalar-inequality} over the
remaining part of the domain.
\end{proof}

\begin{proof}[Proof of Theorem~\ref{pbt-main}]
The eigenvalues of $Y$ lie in $[0,a^{-2}]$.
We apply Lemma~\ref{pbt-scalar-lemma} to each eigenvalue
and take the normalised trace. Since $\tau(Y)=1$,
Lemma~\ref{pbt-moments} gives
$\tau(h_k(Y))\geq1-\tau((Y-I)^2)/2
=1-(d^2-1)/(2M)$.
For $M\geq d^2-1$, this lower bound is nonnegative.
Squaring and using~\eqref{pbt-fidelity-bound} gives
\begin{equation*}
    F_{\mathrm e}(\mathcal E)
    \geq F_{\mathrm s}
    \geq\left(1-\frac{d^2-1}{2M}\right)^2.
\end{equation*}
This proves~\eqref{pbt-main-fidelity} for the deterministic
protocol. The remaining measurement results are included
in $\mathcal E$ and contribute nonnegatively.

The supporting circuits in
Section~\ref{pbt-supporting-circuits} and the complete
count in Section~\ref{pbt-depth-section} establish the
depth bound, including resource preparation and output
selection.
\end{proof}

\subsection{Circuit implementations and depth bounds}
\label{pbt-supporting-circuits}

We now give the circuit bounds needed to complete the proof
of Theorem~\ref{pbt-main}. We first include fanout in the gate
set and count the preparation, permutation and reflection
circuits. We then implement fanout exactly using only
single-qubit and generalised Toffoli gates. All depth counts
include the single-qubit layers.

\subsubsection{Preparation and permutation circuits}
\label{pbt-preparation-depths}

\begin{lem}
\label{pbt-basic-depth}
With fanout included in the gate set, every permutation of
computational basis states has an exact implementation of
depth at most $20$. Every known state with nonnegative
computational basis coefficients can be prepared exactly
in depth at most $29$. Both constructions return all
additional qubits to zero.
\end{lem}

\begin{proof}
Lemma~\ref{permutation to indicator in constant depth}
implements the indicator encoding in depth at most ten.
We follow it by the inverse encoding, with the indicator
positions arranged according to the desired permutation.
The two circuits use the same indicator register.
Their composition implements the permutation in depth
at most twenty and returns all additional qubits to zero.

For state preparation, we use the six steps in
Theorem~\ref{Any state in constant depth}.
Nonnegative coefficients require no final phase restoration.
The rotation parameters also allow zero coefficients,
so no relabelling is needed. The total depth is
$1+5+4+5+4+10=29$, and every qubit outside the output
register returns to zero.
\end{proof}

For the conditional preparation in
Lemma~\ref{pbt-conditional-preparation}, each rotation stage
is implemented by a diagonal unitary between two layers
of single-qubit basis changes. Each stage therefore has
depth $1+7+1=9$. The remaining stages have the same costs
as in the original preparation construction:
\begin{center}
\begin{tabular}{lr}
    Stage & Depth\\
    \hline
    Prepare the product of rotations & $9$\\
    Compute the prefix values & $5$\\
    Apply the conditional inverse rotations & $9$\\
    Reverse the prefix computation & $5$\\
    Add the indicator for the all-zero string & $4$\\
    Decode the indicators & $10$\\
    \hline
    Total & $42$
\end{tabular}
\end{center}

\subsubsection{Conditional rotations and reflections}
\label{pbt-conditional-operations}

We next give the circuits for the conditional rotations
used below and the reflections used in amplitude
amplification. These circuits use only single-qubit
and generalised Toffoli gates.

Write a real single-qubit rotation as
\begin{equation*}
    R(\theta)=
    \begin{pmatrix}
       \cos\theta&-\sin\theta\\
       \sin\theta&\cos\theta
    \end{pmatrix}.
\end{equation*}
To apply $R(\theta)$ when a control qubit is 1, apply
$R(\theta/2)$, CNOT, $R(-\theta/2)$ and CNOT to the
target in that order. For control 0 the rotations cancel.
For control 1 the operation is
$XR(-\theta/2)XR(\theta/2)=R(\theta)$, since
$XR(\alpha)X=R(-\alpha)$. This gives depth four.

To flip a separate target when a register is zero,
negate the tested qubits, apply a generalised Toffoli
and restore the tested qubits. This gives depth three.

A phase of $-1$ on a specified pattern of computational
basis values also has depth at most three. First negate
the qubits whose specified value is zero. Conjugating
a generalised Toffoli by Hadamards on its target then
applies phase $-1$ to the all-one state of the tested
qubits. Finally restore the negated qubits. The Hadamards
and negations combine into one single-qubit layer before
the Toffoli and one after it. For a single tested qubit,
the required phase is a single-qubit gate.

This implements the reflections
$R_{\mathrm{in}}$ and $R_{\mathrm s}$ in
\eqref{pbt-amplification}. Each reflection tests all the
qubits specified in its definition, including the
preparation work qubits. No additional qubit is needed
to store the result of the test.

\subsubsection{One exact amplification step}

The fanout construction below requires the exact
preparation of an auxiliary state. We obtain it using
the following case of amplitude
amplification~\cite{pbt-brassard2002}, in which the
initial success probability is exactly $1/4$ for every
allowed input.

\begin{lem}
\label{pbt-exact-amplification}
Let $V$ be unitary and let $\Pi_{\mathrm{in}}$ and
$\Pi_{\mathrm s}$ project onto the allowed input and
successful output subspaces. Write
$A=\Pi_{\mathrm s}V\Pi_{\mathrm{in}}$.
If $A^\dagger A=\Pi_{\mathrm{in}}/4$, then
\begin{equation}
    -V(I-2\Pi_{\mathrm{in}})V^\dagger
       (I-2\Pi_{\mathrm s})V\Pi_{\mathrm{in}}=2A.
    \label{pbt-exact-amplification-equation}
\end{equation}
\end{lem}

\begin{proof}
Expanding the reflections gives
$V\Pi_{\mathrm{in}}+2A-4VA^\dagger A$.
Substituting $A^\dagger A=\Pi_{\mathrm{in}}/4$
reduces this expression to $2A$.
\end{proof}

Thus one amplification step produces the normalised
successful component with certainty. In the application
below, both projectors specify computational basis
values, so their reflections have the implementations
just described. The overall minus sign can be absorbed
into a single-qubit gate already present in the circuit.

\subsubsection{Exact implementation of fanout}
\label{pbt-fanout-circuit}

We first prepare an auxiliary state and use it to compute
parity. Conjugation by Hadamards then converts parity to
fanout. For the preparation, we use the approximate
circuit of Rosenthal~\cite[Theorem~1.1]
{rosenthal2020boundsqac0complexityapproximating}
and convert it into an exact circuit using
Lemma~\ref{pbt-exact-amplification}.

\begin{lem}
\label{pbt-fanout-depth}
Fanout with any finite number of targets can be
implemented exactly in depth at most $280$ using
single-qubit and generalised Toffoli gates.
Every ancillary qubit starts and ends in zero.
\end{lem}

\begin{proof}
\emph{Preparing the auxiliary state.}
For $t\geq1$, we require a state of the form
\begin{equation}
    \ket{\nu}
    =\frac1{\sqrt2}
       \bigl(\ket{0^t}\ket{\xi_0}
             +\ket{1^t}\ket{\xi_1}\bigr),
    \label{pbt-balanced-state}
\end{equation}
where $\ket{\xi_0}$ and $\ket{\xi_1}$ are normalised
states determined by the construction. The two terms
are orthogonal and have equal weight. Changing their
relative sign thus gives a state orthogonal to
$\ket{\nu}$. We will use this property to distinguish
even and odd parity.

The cited theorem gives a circuit $C$ whose output has
squared overlap at least $1-1/64$ with some state of
the form~\eqref{pbt-balanced-state}. It uses two
multiqubit layers. At most three single-qubit layers
are needed before, between and after them, giving
depth at most five in our convention.

Let $p_0,p_1$ be the probabilities that the first $t$
qubits of $C\ket{0}$ are all zero and all one.
The trace distance from the approximated state is
at most $1/8$, so $p_0,p_1\geq1/2-1/8=3/8$.
We choose $\ket{\xi_b}$ to be the normalised state
of the remaining qubits in the component where
the first $t$ qubits are $b^t$. This component is
$\sqrt{p_b}\ket{b^t}\ket{\xi_b}$.

We use three additional qubits, initially zero,
to give these two components equal amplitudes.
The first two qubits record whether the first $t$
qubits are all zero or all one. These tests take
three layers and one layer respectively.

For each $b\in\{0,1\}$, we rotate the third qubit
conditional on the test for $b^t$, giving amplitude
$1/\sqrt{8p_b}$ on $\ket{1}$. These amplitudes are
at most one because $p_b\geq3/8$. The two patterns
are mutually exclusive, so only one of the rotations
acts on each corresponding component. Each conditional
rotation takes four layers, and we apply them in
sequence. We then reverse the two tests, returning
the first two additional qubits to zero.

Let the resulting circuit be $V_0$. Including $C$,
its depth is at most $5+4+8+4=21$. In the component
where the third additional qubit is 1, each of the
two terms has amplitude
$\sqrt{p_b}/\sqrt{8p_b}=1/\sqrt8$.
The total probability of this result is exactly $1/4$,
and the normalised state of the original registers
has the form~\eqref{pbt-balanced-state}.

We apply Lemma~\ref{pbt-exact-amplification} with the
initial all-zero state as the allowed input and value 1
of the third additional qubit as success.
The reflection about the initial state has depth three.
The reflection for success is a single-qubit phase gate.
The resulting unitary $U$ prepares the required state
exactly in depth at most $3\cdot21+3+1=67$.
We include the additional qubits, whose values are
now fixed, in the states $\ket{\xi_0}$ and
$\ket{\xi_1}$.

The rotation angles depend on $p_0,p_1$ and are chosen
as part of the circuit specification. As in our circuit
conventions, we do not require an efficient classical
procedure for finding these angles.

\emph{Computing parity and restoring the work qubits.}
Let $x_1,\ldots,x_t$ be the input bits and let
$z_1,\ldots,z_t$ be the first $t$ qubits of the
auxiliary state. Starting with all work qubits zero,
we apply $U$, then a controlled-$Z$ gate to each
pair $x_i,z_i$ in parallel, followed by $U^\dagger$.
Call this circuit $W$.

A controlled-$Z$ gate applies phase $-1$ to
$\ket{11}$ and leaves the other computational basis
states unchanged. It is a CNOT conjugated by Hadamards
on its target, so the parallel layer of these gates
has depth three. Hence $W$ has depth at most
$2\cdot67+3=137$.

For a fixed input string $x$, these gates leave the
first term of~\eqref{pbt-balanced-state} unchanged
and multiply the second by $(-1)^{x_1+\cdots+x_t}$.
For even parity, $U^\dagger$ returns all work qubits
to zero. For odd parity, the state before $U^\dagger$
is orthogonal to $\ket{\nu}$, so the work qubits
after $U^\dagger$ are orthogonal to the all-zero state.
Also, $W^2=I$, since each controlled-$Z$ gate squares
to the identity.

We now write the parity into a separate target qubit.
Apply $W$, flip the target when the entire work register
is zero, apply $W$ again, and negate the target.
For even parity, the two target flips cancel.
For odd parity, only the final NOT acts.
The second application of $W$ restores all work qubits
to zero in both cases.

The input bits are unchanged, and the target is
replaced by its XOR with their parity. Linearity
gives the same action on arbitrary quantum inputs.
The depth is at most $2\cdot137+3+1=278$.

\emph{Converting parity to fanout.}
The parity gate is the product of CNOT gates from
the input bits to the target. Conjugating every
data qubit by a Hadamard reverses each CNOT.
The original target becomes the fanout control,
and the other data qubits become its targets.
These two Hadamard layers give depth at most $280$.
All work qubits still return to zero.
Fanout with no targets is the identity.
\end{proof}

The bound is independent of the number of targets.
To replace several fanouts in one layer, we assign
disjoint work registers to their implementations
and run them in parallel. Each replacement returns
its work qubits to zero. Replacing fanout throughout
a circuit thus multiplies its depth by at most $280$.

\subsection{Total depth and dependence on the input dimension}
\label{pbt-depth-section}

We now combine the preceding bounds to count the complete
teleportation circuit, including preparation of the shared
pairs and selection of the output port. We first include
fanout in the gate set. Circuit size and the number of
ancillary qubits remain unrestricted.

The circuit $V$ contains the first conditional preparation,
the inverse of the second, and one single-qubit rotation.
Its depth is at most $42+42+1=85$.
Each reflection in~\eqref{pbt-amplification} has depth
at most three. Since $q=(k-1)/2$, the circuit $U_q$
contains $2q+1=k$ applications of $V$ or $V^\dagger$
and $2q=k-1$ reflections.

Each shared pair has nonnegative computational basis
coefficients. By Lemma~\ref{pbt-basic-depth}, all $M$
pairs can be prepared in parallel in depth at most $29$,
using separate work qubits for each preparation.

We use $B_1$ as the fixed output register.
When the success condition holds and the label specifies
a valid port $i$, we swap $B_i$ with $B_1$.
For every other result, we leave Bob's registers unchanged,
so $B_1$ is again the selected output.
The operation preserves Alice's registers and permutes
computational basis states. Lemma~\ref{pbt-basic-depth}
implements this entire permutation in depth at most $20$.

The depth count is
\begin{center}
\begin{tabular}{lr}
    Operation & Depth\\
    \hline
    Prepare the $M$ shared pairs & $29$\\
    Apply $V$ or $V^\dagger$ a total of $k$ times & $85k$\\
    Apply the $k-1$ reflections & $3(k-1)$\\
    Select the output port & $20$\\
    \hline
    Total & $88k+46$
\end{tabular}
\end{center}

We next replace fanout using
Lemma~\ref{pbt-fanout-depth}.
Gates in each original layer have disjoint supports,
so their replacements can run in parallel with separate
work registers. The complete circuit then has depth at most
$280(88k+46)$
using only single-qubit and generalised Toffoli gates.
Every replacement returns its work qubits to zero,
so the subsequent operations, including the reflections,
act as before.

For example, when $d=2$, we can take $k=3$.
The depth bound is then $310$ with fanout included
in the gate set and $86800$ using only single-qubit
and generalised Toffoli gates. These are upper bounds
obtained by composing the preceding circuits.
We have not optimised the constants.

For $0<\varepsilon<1$, taking
$M=\lceil(d^2-1)/\varepsilon\rceil$ gives
entanglement fidelity at least $1-\varepsilon$.
Increasing $M$ improves this fidelity bound without
changing $k$ or the depth bound. The preparation circuits
and their gate parameters may change with $M$.

The shared pairs occupy $2M\lceil\log_2d\rceil$ qubits.
The circuit also uses work qubits for the preparations,
permutations and fanout implementations. Their number
and the total gate count are unrestricted in this result.
All circuit operations are exact, while teleportation
is approximate for finite $M$.

The dependence on $d$ comes from the number of
amplification steps. For the allowed values of $M$,
the initial success probability satisfies
$1/(d+1)\leq\mu<1/d$. Our choice of $k$ follows
$k<\pi\sqrt{d+1}/2+2$, giving depth $O(\sqrt d)$
independently of $M$.

Thus every fixed input dimension admits a depth bound
that is independent of the desired accuracy.
For an input of $n$ qubits, $d=2^n$ and the bound
is $O(2^{n/2})$. Whether one depth bound can hold
for all input dimensions, even at a fixed positive
entanglement fidelity, remains open. 

\section*{Acknowledgements}
Strelchuk acknowledges support from the Wellcome Leap as part of the Q4Bio Program and the Royal Society University Research Fellowship. Subramanian acknowledges support from the Royal Society through a University Research Fellowship.

\clearpage

\printbibliography
\end{document}